\documentclass[draftclsnofoot,onecolumn]{IEEEtran}

\usepackage{mathtools,bm}
\usepackage{hyperref}
\usepackage{microtype}
\usepackage{booktabs}
\usepackage{cite}
\usepackage[nameinlink,capitalise]{cleveref}
\usepackage{listings}

\usepackage{enumitem}
\usepackage{array}
\usepackage{multirow}
\usepackage{makecell}
\usepackage[export]{adjustbox}
\usepackage{amsthm,amsmath,amssymb}
\usepackage{mathrsfs}
\usepackage{balance}
\usepackage{algorithm,algorithmic}

\hypersetup{
	colorlinks=true,
	linkcolor=blue,
	citecolor=blue,
	urlcolor=blue
}

\usepackage{orcidlink}
\newcounter{magicrownumbers}

\newtheorem{theorem}{Theorem}
\newtheorem{lemma}{Lemma}
\newtheorem{proposition}{Proposition}
\newtheorem{corollary}{Corollary}
\newtheorem{assumption}{Assumption}
\newtheorem{remark}{Remark}

\newcommand{\cX}{\mathcal X}
\newcommand{\cC}{\mathcal C}
\newcommand{\cS}{\mathcal S}
\newcommand{\cV}{\mathcal V}
\newcommand{\cR}{\mathcal R}
\newcommand{\cY}{\mathcal Y}
\newcommand{\cZ}{\mathcal Z}
\newcommand{\cP}{\mathcal P}
\newcommand{\cM}{\mathcal M}
\newcommand{\cL}{\mathcal L}
\newcommand{\cU}{\mathcal U}
\newcommand{\cQ}{\mathcal Q}
\newcommand{\bbP}{\mathbb P}

\newcommand{\bbR}{\mathbb R}
\newcommand{\bbZ}{\mathbb Z}

\newcommand{\vol}{\operatorname{vol}}
\newcommand{\aff}{\operatorname{aff}}
\newcommand{\conv}{\operatorname{conv}}

\newcommand{\diag}{\operatorname{diag}} 
\newcommand{\relbd}{\partial_{\mathrm{rel}}} 
\usepackage{caption}
\title{
	Volume-Refined Achievability and Converse Approximations for Noisy Permutation Channels
}
\author{
	Lugaoze Feng$^{\orcidlink{0009-0000-4014-4154}}$,
	Guocheng Lv$^{\orcidlink{0000-0002-7136-3402}}$,
	Xunan Li$^{\orcidlink{0000-0002-5740-161X}}$, and
	Ye Jin
	\thanks{Lugaoze Feng, Guocheng Lv and Ye Jin are with the School of Electronics, Peking University, Beijing 100871, China (e-mail: lgzf@stu.pku.edu.cn; lv.guocheng@pku.edu.cn; jinye@pku.edu.cn).}
	\thanks{Xunan Li is with the National Computer Network Emergency Response Technical Team/Coordination Center of China, Beijing 100029, China (e-mail: lixunan@cert.org.cn).}}

\date{}

\begin{document}
	\maketitle

	\begin{abstract}
	We study volume-refined achievability and converse bounds for noisy permutation channels
	generated by strictly positive DMCs, allowing the reachable
	output polytope to have arbitrary affine dimension \(d\ge1\).  The reachable
	output polytope may be lower-dimensional than the output simplex, whereas
	existing refined achievability analyses and fixed-error converses are not
	adapted to this intrinsic affine geometry. 
	On the achievability side, we develop an affine-coordinate simplex-lattice
	construction adapted to the reachable output polytope, together with a
	nearest-neighbor decoder and a geometric error-reduction argument in the same
	coordinate space.  This yields a Gaussian achievability approximation with an
	\(o(1)\) remainder.
	On the converse side, we first use a meta-converse combined with a KL covering
	and a local testing estimate to obtain a fixed-error converse with a bounded
	remainder, which implies the logarithmic \(\epsilon\)-capacity \(d/2\).  We then
	apply the meta-converse with a stratified Jeffreys-mixture auxiliary output
	distribution.  Using a local Laplace approximation and a local likelihood-ratio approximation, this choice identifies the Fisher-volume term and an explicit
	Gaussian testing constant, yielding a constant-order converse approximation
	with an \(o(1)\) remainder. The achievability and converse constants arise from different constructions and are not claimed to match in general.
	\end{abstract}


	\section{Introduction}
	The \textit{noisy permutation channel} model describes a setting in which transmitted data packets arrive corrupted and out of order. It is a useful abstraction for \textit{multipath routing networks}, where packets may experience unpredictable latency variations due to changes in network topology or load-balancing-induced rerouting \cite{maclaren_walsh_optimal_2009}. This setting has been extensively studied in the coding literature \cite{kovacevic_subset_2013,kovacevic_codes_2018,kovacevic_perfect_2015}. The noisy permutation channel is also relevant to \textit{DNA storage systems} \cite{yazdi_dna-based_2015,kiah2016codes}, where DNA sequences are broken during storage and read using order-independent shotgun sequencing. The study of binary noisy permutation channels was initiated in \cite{makur_bounds_2020} and extended to general noisy permutation channels in \cite{makur_coding_2020}. The capacity result was further investigated by \cite{tang_capacity_2023} using \textit{divergence covering}. Additionally, this channel setting was extended to adder multiple-access channels and more general  multiple-access channels in \cite{lu2024permutation_add} and \cite{lu2024permutation}, respectively.
	
	The systems to which noisy permutation channels are applicable typically operate
	at blocklengths ranging from hundreds to thousands of symbols
	\cite{heckel_fundamental_2017,maclaren_walsh_optimal_2009}, where first-order
	asymptotic results may not accurately describe the finite-blocklength behavior.
	Several finite-blocklength and refined-asymptotic results have therefore been
	developed for noisy permutation channels. The Gaussian achievability analyses
	in \cite{feng_lower_isit_2025} were developed primarily
	for strictly positive channels whose reachable output polytope is full-dimensional relative to the output simplex, and their code construction uses grids in the output simplex. On the converse side,
	\cite{feng_upper_tcomm_2025} developed finite-blocklength upper bounds based on
	\textit{binary hypothesis testing}, symbol relaxation, and Kullback-Leibler (KL)
	divergence covering.

	These results leave open the refined finite-blocklength behavior when the reachable output distributions are contained in a lower-dimensional affine slice of the output simplex. In the achievability analysis of \cite{feng_lower_isit_2025}, the code construction is based on a uniform grid in the full output simplex \(\Delta_{|\mathcal Y|-1}\), and is therefore tailored to channels whose reachable output polytope is full-dimensional relative to that simplex. In the general noisy permutation channel model, however, the reachable output polytope may occupy only a lower-dimensional affine slice of the output simplex, and only distributions in this polytope can be induced by input distributions.
	A full-simplex grid is then not naturally aligned with the reachable affine hull and does not directly provide a regular collection of reachable probability points, whereas a direct construction on an arbitrary lower-dimensional polytope would have to account for its possible shape and lattice structure.
	Thus, the full-dimensional setting considered in the existing refined random-coding analysis should be viewed as a special case, and a different construction is needed to obtain a uniform refined achievability analysis in the general affine-geometric setting.
	
	On the converse side, the mutual-information covering approach of
	\cite{tang_capacity_2023} identifies the first-order logarithmic capacity for
	noisy permutation channels generated by a strictly positive discrete
	memoryless channel (DMC) with transition matrix \(W\). The logarithmic capacity
	equals one half of the affine dimension of the reachable output polytope,
	equivalently \((\operatorname{rank}(W)-1)/2\). However, this approach yields only a weak converse and therefore does not directly address fixed error probabilities.
	Existing fixed-error bounds \cite{feng_upper_tcomm_2025}, on the
	other hand, are based on a \textit{meta-converse}
	\cite{polyanskiy_channel_2010} combined with a symbol-relaxation
	step. This relaxation avoids a direct local comparison between the
	noisy permutation transition kernel and product auxiliary output
	distributions, but the resulting bounds are governed by the dimension of the
	output simplex \( |\mathcal Y|-1 \), rather than by the affine dimension
	\(d\) of the reachable polytope. In the lower-dimensional setting
	\(d<|\mathcal Y|-1\), this incurs an additional
	\[
	(|\mathcal Y|-1-d)\log\sqrt n
	\]
	term and does not yield a bounded-remainder fixed-error converse
	governed by the reachable affine geometry. The key step is therefore to
	retain the noisy permutation transition kernel in the binary test and control
	the local test between this kernel and a nearby product auxiliary distribution,
	uniformly in the blocklength. Establishing this local hypothesis-testing
	estimate is a central step in obtaining a fixed-error converse governed by the
	affine dimension \(d\).

	These considerations motivate an affine-coordinate achievability construction
	with a corresponding error-reduction argument, as well as a fixed-error
	converse that retains the noisy permutation transition kernel and controls
	the resulting local testing penalty uniformly in the blocklength.

	This paper develops refined achievability and converse bounds whose
	blocklength-dependent term is governed by the affine dimension of the reachable
	output polytope. To this end, we make the geometry of the reachable output set explicit.
	The set of reachable output distributions is the \textit{output polytope}
	\[
	\mathcal P_W
	:=
	\operatorname{conv}\{W(\cdot|x):x\in\mathcal X\},
	\]
	where \(W\) denotes the DMC transition matrix.
	Since the random permutation block removes ordering information, the empirical output distribution, or equivalently the output \textit{multiset}, is a sufficient statistic for decoding. We therefore formulate our achievability and converse bounds in terms of the affine dimension 
	\[ 
		d:=\dim \mathcal P_W . 
	\] 
	Equivalently, since the rows of \(W\) are probability vectors, \(d=\operatorname{rank}(W)-1\). We use the geometric notation \(\dim\mathcal P_W\) throughout because the code construction, the covering argument, and the volume factor are all formulated on the affine hull of the reachable output polytope.

	On the achievability side, the proposed simplex-lattice construction yields,
	for every \(c<c_\epsilon\),
	\[
	\log M^\star(n,\epsilon)
	\ge
	d\log(c\sqrt n)+\log \lambda_W^\star-\log d!+o(1),
	\]
	where \(c_\epsilon\) is the Gaussian coefficient associated with this
	construction and is determined by an average of local coordinate variances over
	the reachable output polytope.  The factor \(\lambda_W^\star\) is the Euclidean
	volume ratio induced by the chosen minimum-volume reference simplex in the
	reachable affine hull, and it enters the achievability bound through the
	simplex-lattice counting estimate.
	
	On the converse side, the refined Jeffreys-mixture construction yields
	\[
	\log M^\star(n,\epsilon)
	\le
	d\log \sqrt{n}+\log\mathcal J_W+C_{W,\epsilon}^{J,{\rm str}}+o(1),
	\]
	where \(\mathcal J_W\) is the Fisher volume of the reachable output polytope
	and \(C_{W,\epsilon}^{J,{\rm str}}\) can be evaluated from the Gaussian limit experiment. The achievability and converse approximations both have the
	same affine-dimensional blocklength term \(d\log\sqrt n\), but their
	constant-order geometric factors arise from different natural constructions:
	the Euclidean reference-simplex geometry on the achievability side and the
	Fisher-volume geometry on the converse side. The constants in the achievability and converse bounds need not coincide.

	Our main contributions are as follows:
	\begin{enumerate}
		\item[$\bullet$]
		We establish an achievability bound for noisy permutation channels generated by
		strictly positive DMCs with arbitrary reachable affine dimension \(d\ge1\).  The construction
		represents the reachable output polytope by its affine preimage, a
		full-dimensional coordinate polytope, so that simplex-lattice counting and
		nearest-neighbor error reduction can be applied uniformly.  A geometric
		argument shows that each decoding error is contained in a union of finitely
		many one-dimensional transfer events in the affine coordinate space.
		
		\item[$\bullet$]
		We develop a local hypothesis-testing estimate for the meta-converse with the
		noisy permutation transition kernel retained in the binary test.  Combined with
		an affine-dimensional KL covering of \(\mathcal P_W\), this gives a
		bounded-remainder fixed-error converse.  Together with the achievability bound,
		it yields the logarithmic \(\epsilon\)-capacity \(d/2\).
		
		\item[$\bullet$]
		We refine the Gaussian achievability analysis by retaining the local
		coordinate variances over the reachable output polytope.  This yields, for the
		proposed simplex-lattice construction, a Gaussian achievability approximation
		with a construction-dependent coefficient \(c_\epsilon\) determined by an
		average over the reachable output polytope.
		
		\item[$\bullet$]
		We refine the converse approximation by applying the meta-converse with a
		stratified Jeffreys-mixture auxiliary distribution.  Its \(\mathcal P_W\)
		component identifies, through a local Laplace approximation, the Fisher-volume
		term \(\log\mathcal J_W\) for interior output types, while the local
		likelihood-ratio expansion yields an explicit Gaussian testing constant.  Its
		lower-dimensional face components control boundary output types and make the
		bound uniform.  This gives, with an \(o(1)\) remainder, a constant-order converse upper approximation parallel to the Gaussian achievability lower approximation.
	\end{enumerate}
	The remainder of this section introduces notation. Section~\ref{sec:system_model}
	presents the system model and defines the reachable output polytope.
	Section~\ref{sec:achievability} develops the affine-coordinate
	simplex-lattice achievability bound. Section~\ref{sec:converse-covering}
	establishes the converse bound and the logarithmic \(\epsilon\)-capacity through KL covering . 
	Section~\ref{sec:refined-average-gaussian}
	develops the Gaussian achievability approximation. Section~\ref{sec:jeffreys-mixture-converse} develops the stratified
	Jeffreys-mixture refined converse, and
	Section~\ref{sec:numerical_results} presents numerical results.
	
	\subsection{Notation}
	\label{sec:notation}
	
	We use standard information-theoretic notation throughout.  The finite input and output alphabets are denoted by $\cX$ and $\cY$, with cardinalities $|\cX|=k$ and $|\cY|=m$.  For a positive integer $s$, let
	\[
	\Delta_{s-1}:=\left\{p\in\bbR^s: p_i \ge 0,\sum_{i=1}^s p_i=1\right\}
	\]
	denote the standard probability simplex.  All logarithms are base two unless explicitly stated otherwise. For distributions \(P\) and \(Q\) on the same finite alphabet, the KL divergence is
	\[
	D(P\|Q):=\sum_z P(z)\log\frac{P(z)}{Q(z)} .
	\] 
	For \(y^n\in\cY^n\), let \(N(y^n)=(N_y(y^n))_{y\in\cY}\), where
	\(N_y(y^n):=\sum_{t=1}^n \mathbf 1\{y_t=y\}\).  For
	\(q\in\Delta_{m-1}\), \(\operatorname{Mult}(n,q)\) denotes the multinomial
	distribution of \(N(Y^n)\) when \(Y_1,\ldots,Y_n\) are independent with common
	law \(q\).
	
	For two distributions \(P\) and \(Q\) on the same space
	\(\cZ\), and for \(\alpha\in(0,1)\), the Neyman-Pearson type-II error with
	power at least \(\alpha\) under \(P\) is defined as
	\begin{equation}
		\beta_\alpha(P,Q)
		:=
		\inf_{P_{T|Z}:\,P[T=1]\ge \alpha} Q[T=1],
	\end{equation}
	where the infimum is over all random transformations $P_{T|Z}:\cZ \to \{0,1\}$. The event \(T=1\)
	means that the test chooses \(P\). 
	
	For a set \(B\subseteq\bbR^s\), let \(\operatorname{span}(B)\),
	\(\conv(B)\), \(\aff(B)\), and \(\dim B\) denote its linear span, convex hull,
	affine hull, and affine dimension, respectively.  If \(B\) is contained in a \(d\)-dimensional affine subspace,
	then \(\vol_d(B)\) denotes the \(d\)-dimensional Euclidean volume induced on
	that affine subspace.  
	
	For a convex set \(B\), \(\operatorname{relint}(B)\) denotes the interior
	relative to \(\aff(B)\), and \(\relbd B\) denotes the boundary of \(B\) in the
	relative topology of \(\aff(B)\).  If \(B\) is closed in \(\aff(B)\), then
	\[
	\relbd B=B\setminus\operatorname{relint}(B).
	\]
	More generally, a \(d\)-simplex means the convex hull of \(d+1\) affinely
	independent points.
	
	Unless otherwise specified, \(\|\cdot\|\) denotes the Euclidean norm and
	\[
	\operatorname{dist}(x,B):=\inf_{b\in B}\|x-b\|
	\]
	denotes Euclidean distance to a set \(B\).  When \(B\) lies in a lower-dimensional
	affine subspace, this is the Euclidean distance inherited from the ambient
	space; for \(x\in\aff(B)\), it is equivalently the distance computed within
	\(\aff(B)\).

	Let \(e_i^{(s)}\) denote the \(i\)-th standard basis vector of \(\mathbb R^s\).
	When the dimension is clear from context, we write \(e_i\).  Output-alphabet
	basis vectors in \(\mathbb R^m\) will be denoted separately by \(b_y\). The notation \(O_W(1)\) denotes a quantity uniformly bounded, for all
	sufficiently large \(n\), by a finite constant depending only on \(W\).
	Similarly, \(O_{W,\epsilon}(1)\) allows the constant and the threshold on \(n\) to depend also on the fixed error probability \(\epsilon\). The standard normal cumulative distribution function is denoted by $\Phi$, and $\Phi^{-1}$ denotes its inverse.
	
	For an input sequence $x^n=(x_1,\ldots,x_n)\in\cX^n$, define the product output distribution
	\begin{equation}
		W_{x^n}(y^n):=\prod_{t=1}^n W(y_t|x_t),
		\qquad y^n\in\cY^n.
		\label{eq:def-Wxn}
	\end{equation}
	If \(Q_1,\ldots,Q_n\) are output distributions on \(\cY\), then
	\(\prod_{t=1}^n Q_t\) denotes the product distribution on \(\cY^n\) given by
	\[
	\left(\prod_{t=1}^n Q_t\right)(y^n)
	=
	\prod_{t=1}^n Q_t(y_t).
	\]
	When \(Q_1=\cdots=Q_n=Q\), we write this product distribution as
	\(Q^{\otimes n}\).

	\section{System Model and Reachable Output Polytope} \label{sec:system_model}

	\subsection{System Model}
	
	A code \(\cC_n\) consists of a message set \(\cM\), an encoding function
	\(f_n:\cM\to\cX^n\), and a decoding function
	\(g_n:\cY^n\to\cM\cup\{\mathtt e\}\). The message \(U\) is uniformly distributed
	on \(\cM\), and the input codeword is \(X^n=f_n(U)\).
	
	Let $\cX$ and $\cY$ be finite input and output alphabets with $|\cX|=k$ and $|\cY|=m$. 
	A DMC from $\cX$ to $\cY$ is represented by a row-stochastic matrix
	\[
	W=\{W(y|x):x\in\cX,\ y\in\cY\}.
	\] 
	For \(p \in \Delta_{k-1}\), we denote the induced output distribution by
	\[
	q :=p W,
	\qquad
	q (y)=\sum_{x\in\cX}p (x)W(y|x),
	\quad y\in\cY.
	\]
	
	The transmitter uses the encoder to map the message to an \textit{input codeword} \(X^n\). 
	\(X^n\) goes through the DMC \(W:\cX\to\cY\) to generate a \textit{noisy codeword} \( Z^n\). Then, \(Z^n\) goes through an independent uniformly random permutation block
	\(P_{Y^n|Z^n}\) to produce the \textit{output codeword} \(Y^n\), where
	\(P_{Y^n|Z^n}\) is defined as follows. Let \(\sigma:\{1,\ldots,n\}\to\{1,\ldots,n\}\) be drawn uniformly and randomly from the symmetric group \(\cS_n\) over \( \{ 1,...,n \}\). Then \(Y_{\sigma(i)}=Z_i\) for \(i=1,\ldots,n\). This system is shown in Fig.~\ref{sys:perm}.
	
	\begin{figure*}[t]
		\normalsize	
		\centering
		 \includegraphics[width = 1\textwidth]{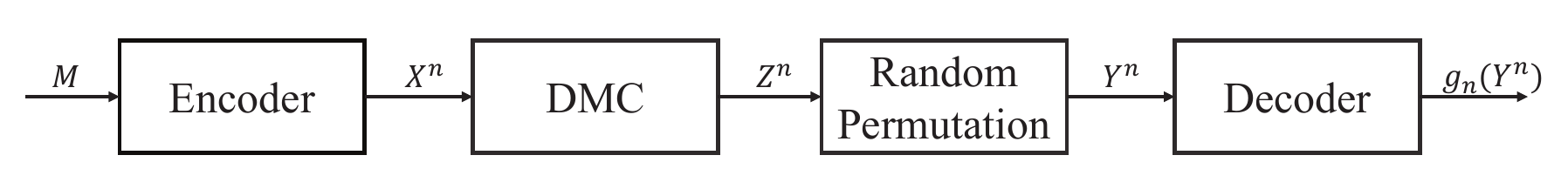}
		 \captionsetup{font=footnotesize}
		 \caption{\  Illustration of a communication system with a DMC followed by a random permutation block.}
		 \label{sys:perm}
	\end{figure*}
	
	The average error probability of code \(\cC_n\) is
	\(
	P_e:=\bbP[g_n(Y^n)\ne U].
	\)
	Since the permutation is uniform and independent, the channel transition law
	is invariant under permutations of the output coordinates. Therefore, any
	decoder can be symmetrized over output permutations without increasing the
	average error probability. The empirical output distribution, or
	equivalently the output count vector, is a sufficient statistic for decoding in
	the average error setting considered here.
	For a given
	\(\epsilon\in(0,1)\) and blocklength \(n\), we define
	\[
	M^\star(n,\epsilon)
	:=
	\max\bigl\{
	|\cM|:\exists\,\cC_n \text{ s.t. } P_e\le\epsilon
	\bigr\}.
	\]
	Since the random permutation removes ordering information, the number of
	possible output empirical distributions is polynomial in the blocklength
	\(n\).  We therefore use the logarithmic normalization \(\log M/\log n\).
	For fixed \(\epsilon\in(0,1)\), define the logarithmic \(\epsilon\)-capacity by
	\begin{equation}
		C_\epsilon
		:=
		\liminf_{n\to\infty}
		\frac{\log M^\star(n,\epsilon)}{\log n}. \nonumber
	\end{equation}
	The logarithmic capacity is then defined as
	\begin{equation}
		C := \lim_{\epsilon \to 0^+} C_\epsilon. \nonumber
	\end{equation}
	
	\subsection{Reachable Output Geometry}
	
	The set of all reachable output distributions is the convex polytope
	\begin{equation}
		\begin{aligned}
			\cP_W
			&:=\{pW:p\in\Delta_{k-1}\}\\
			&=\conv\{W(\cdot|x):x\in\cX\}\\
			&\subseteq \Delta_{m-1}.
		\end{aligned}
		\label{eq:reachable-polytope}
	\end{equation}
	Because every row of \(W\) sums to one, \(\cP_W\) lies in the affine hyperplane
	\(\{q\in\bbR^m:\sum_y q(y)=1\}\). Let
	\begin{equation}
		d:=\dim \cP_W
		\label{eq:def-d}
	\end{equation}
	be the affine dimension of the reachable output polytope. If \(r\) denotes the
	row rank of \(W\), then
	\[
	d=r-1,
	\]
	since the rows of \(W\) are probability vectors. Indeed, the rows of \(W\)
	lie in the affine hyperplane \(\{q\in\mathbb R^m:\sum_y q(y)=1\}\), which
	does not pass through the origin; hence their linear span has dimension one
	larger than their affine hull. We call the case
	\(d<m-1\), equivalently \(r<m\), the lower-dimensional setting.
	For the \((d+1)\)-dimensional coordinate space, we index the standard basis vectors by \(0,\ldots,d\) and write \(e_i:=e_i^{(d+1)}\), \(0\le i\le d\).
	
	Recall that the output alphabet is \(\cY=\{1,\ldots,m\}\). For
	\(y\in\cY\), let \(b_y:=e_y^{(m)}\in\mathbb R^m\) denote the standard basis
	vector corresponding to the output symbol \(y\).
	
	In this paper, we impose the following strictly positive matrix assumption.
	\begin{assumption} 
		\label{ass:positive-support}
		Assume that the transition probabilities of \(W\) are strictly positive:
		\[
		p_{\min}:=\min_{x\in\cX,\ y\in\cY} W(y|x)>0.
		\]
		Equivalently, for every \( q\in\cP_W \) and every \(y\in\cY \), 
		\(
		q(y)\ge p_{\min}.
		\)
	\end{assumption}
	\begin{remark}
		Assumption~\ref{ass:positive-support} is used to obtain uniform local
		estimates over the reachable output polytope.  In particular, it yields a
		uniform comparison between KL divergence and squared Euclidean distance and
		prevents the local coordinate variances and multinomial probabilities from
		degenerating near the boundary of the output simplex. 
	\end{remark}
	
	\section{Achievability Bounds}\label{sec:achievability}
	In this section, we prove the achievability bound by constructing codes through an affine-coordinate parametrization of the reachable output polytope. When \(\cP_W\) is lower-dimensional relative to the output simplex, a grid in the full output simplex is not naturally aligned with the reachable output set. We therefore represent \(\cP_W\) by its affine preimage \(K_W\), a full-dimensional coordinate polytope inside the standard simplex, and place a simplex lattice on \(K_W\).
	This ensures that every message point corresponds to a reachable output distribution while preserving a regular lattice structure suitable for cardinality estimates and decoding.
	
	Section~\ref{sec:Simplex-Lattice Message Sets} constructs this simplex-lattice message set and establishes a volume-sensitive cardinality estimate. 
	Section~\ref{sec:Coordinate Decoder} introduces a projection-based nearest-neighbor decoder: the empirical output distribution is first projected onto the reachable affine hull and then expressed in the affine coordinates used for the message set.  Section~\ref{sec:Voronoi-Transfer Reduction} develops the geometric error-reduction argument, showing that any nearest-neighbor decoding error must cross one of finitely many elementary transfer halfspaces. Finally,
	Section~\ref{sec:Low-Rank Gaussian Achievability} combines this reduction with the random-coding construction to derive the finite-blocklength achievability bound.
	
	\subsection{Simplex-Lattice Message Sets} \label{sec:Simplex-Lattice Message Sets}
	Let
	\begin{equation}
		\Delta_d:=\left\{u\in\bbR^{d+1}:u_i\ge0,\ \sum_{i=0}^d u_i=1\right\}
		\label{eq:def-Delta-d}
	\end{equation}
	denote the standard \(d\)-simplex of affine coordinates.
	
	Let
	\begin{equation}
		\cR_W:=\aff(\cP_W)\cap\Delta_{m-1}.
		\label{eq:ambient-output-slice}
	\end{equation}
	This is the portion of the output simplex that lies in the affine hull of the
	reachable output polytope.  Thus, when \(\cP_W\) is not full-dimensional relative to the output simplex 
	\(\Delta_{m-1}\), the natural reference region is \(\cR_W\) rather than \(\Delta_{m-1}\).
	
	Let \(S_W^\star\) be a minimum-volume \(d\)-simplex satisfying
	\[
	\cR_W \subseteq S_W^\star\subseteq \aff(\cP_W).
	\]
	Such a simplex exists because \(\cR_W\) is compact and
	\(d\)-dimensional in \(\aff(\cP_W)\).  Indeed, a minimizing sequence may be
	chosen with uniformly bounded volume; since \(\cR_W\) has nonempty relative
	interior, such a sequence has vertices in a bounded subset of \(\aff(\cP_W)\),
	and compactness yields a convergent subsequence whose limit is an enclosing
	\(d\)-simplex.  Define the volume ratio
	\begin{equation}
		\lambda_W^\star
		:=
		\frac{\vol_d(\cP_W)}
		{\vol_d(S_W^\star)}
		=
		\frac{\vol_d(\cP_W)}
		{
			\inf_{ S:\,\cR_W\subseteq S\subseteq\aff(\cP_W)}
			\vol_d( S)
		}.
		\label{eq:lambdaW}
	\end{equation}
	  All volumes are
	\(d\)-dimensional Euclidean volumes induced on \(\aff(\cP_W)\). When \(\cR_W\) is a \(d\)-simplex, this ratio reduces to
	\[
	\lambda_W^\star
	=
	\frac{\vol_d(\cP_W)}{\vol_d(\cR_W)}.
	\]
	Hence, $\lambda_W^\star$ measures the volume of the reachable
	output polytope normalized by that of the reference simplex.
	Write the vertices of \( S_W^\star\) as \(s_0,\ldots,s_d\). Define the
	affine bijection
	\begin{equation}
		T:\Delta_d\to S_W^\star,
		\qquad
		T(u)=\sum_{i=0}^d u_i s_i.
		\label{eq:T-map}
	\end{equation}
	We use the same notation \(T\) for its affine extension from
	\(\{u\in\bbR^{d+1}:\sum_{i=0}^d u_i=1\}\) onto \(\aff(\cP_W)\), and \(T^{-1}\)
	denotes the inverse of this extension.
	
	The preimage of the reachable output polytope is
	\begin{equation}
		K_W:=T^{-1}(\cP_W)\subseteq\Delta_d.
		\label{eq:KW}
	\end{equation}
	Since \(T\) is an affine bijection from \(\Delta_d\) onto
	\( S_W^\star\), it scales all \(d\)-dimensional volumes by the same
	constant. Therefore,
	\begin{equation}
		\frac{\vol_d(K_W)}{\vol_d(\Delta_d)}
		=
		\frac{\vol_d(\cP_W)}{\vol_d(S_W^\star)}
		=
		\lambda_W^\star .
		\label{eq:KW-volume-ratio}
	\end{equation}
	Clearly, \(K_W\) is a full-dimensional convex polytope contained in \(\Delta_d\).

	For a positive integer \(N\), define the simplex \(N\)-grid, or simplex lattice,
	\begin{equation}
		\cL_N
		:=
		\left\{
		u\in\Delta_d:
		u_i=\frac{a_i}{N},\ a_i\in\mathbb Z_{\ge0},\
		\sum_{i=0}^d a_i=N
		\right\}.
	\end{equation}
	Equivalently, \(\cL_N=\Delta_d\cap N^{-1}\mathbb Z^{d+1}\).  We refer to
	\(N\) as the lattice resolution.
	The message set is then defined as
	\begin{equation}
		\cU_N:=K_W\cap\cL_N.
		\label{eq:UN}
	\end{equation}
	The messages are indexed by lattice points in the affine coordinate representation of the reachable output polytope, rather than by input
	distributions.
	The corresponding set of output distributions is
	\begin{equation}
		\cQ_N:=T(\cU_N)\subseteq\cP_W.
		\label{eq:QN}
	\end{equation}
	Thus each message point \(u\in\cU_N\) is paired with a target output
	distribution \(q_u=T(u)\). Equivalently, \(u=T^{-1}(q_u)\) is the
	affine coordinate representation of \(q_u\) induced by the simplex
	\(S_W^\star\). This construction guarantees reachability: if \(q=T(u)\) with
	\(u\in\cU_N\), then \(u\in K_W=T^{-1}(\cP_W)\), hence \(q\in\cP_W\). Therefore,
	by \eqref{eq:reachable-polytope}, there exists at least one input distribution
	\(p_u\) such that \(p_uW=T(u)\).

	The following lemma estimates the cardinality of the message set.
	
	\begin{lemma}
		\label{lem:lattice-counting}
		As \(N\to\infty\),
		\begin{equation}
			|K_W\cap\cL_N|
			=
			\lambda_W^\star\binom{N+d}{d}
			+
			O_W(N^{d-1}).
			\label{eq:lattice-count-asymp}
		\end{equation}
		Consequently,
		\begin{equation}
			\log |\cU_N|
			=
			d\log N+\log\lambda_W^\star-\log d!+o(1).
			\label{eq:log-cardinality}
		\end{equation}
	\end{lemma}	
	\begin{proof}
		The points of \(\cL_N\) are in one-to-one correspondence with integer vectors
		\((a_0,\ldots,a_d)\in\mathbb Z_{\ge0}^{d+1}\) satisfying
		\(\sum_{i=0}^d a_i=N\). Hence, by the stars-and-bars formula,
		\begin{equation}
			|\cL_N|
			=
			\binom{N+d}{d}
			=
			\frac{N^d}{d!}\left(1+O\left(\frac1N\right)\right).
			\label{eq:stars-and-bars}
		\end{equation}
		
		Let
		\[
		\Lambda
		:=
		\left\{
		z\in\mathbb Z^{d+1}:\sum_{i=0}^d z_i=0
		\right\}
		\]
		be the lattice in the linear span of \(\Delta_d\).  The simplex lattice
		\(\cL_N\) is the intersection of \(\Delta_d\) with a translate of the
		scaled lattice \(N^{-1}\Lambda\).  With respect to the induced
		\(d\)-dimensional Euclidean measure on
		\(\{z:\sum_i z_i=0\}\), the determinant of \(\Lambda\) is \(\sqrt{d+1}\),
		whereas
		\( 
		\vol_d(\Delta_d)=\frac{\sqrt{d+1}}{d!}.
		\)
		Consequently,
		\begin{equation}
			\frac{N^d\vol_d(\Delta_d)}{\det(\Lambda)}
			=
			\frac{N^d}{d!},
			\label{eq:volume_main_term}
		\end{equation}
		which agrees with the leading term of the stars-and-bars count above.

		Since \(K_W\) is a fixed convex polytope with piecewise flat boundary, the
		Lipschitz principle for lattice-point counting \cite{davenport1951principle},
		applied after identifying \(\operatorname{span}(\Delta_d-\Delta_d)\) with
		\(\mathbb R^d\) through a basis of \(\Lambda\), gives
		\[
		|K_W\cap\cL_N|
		=
		\frac{N^d\vol_d(K_W)}{\det(\Lambda)}
		+
		O_W(N^{d-1}).
		\]
		Combining this estimate with \eqref{eq:stars-and-bars} and \eqref{eq:volume_main_term}, we obtain
		\[
		|K_W\cap\cL_N|
		=
		\frac{\vol_d(K_W)}{\vol_d(\Delta_d)}|\cL_N|
		+
		O_W(N^{d-1}).
		\]
		Using \eqref{eq:KW-volume-ratio} and
		\(|\cL_N|=\binom{N+d}{d}\) yields
		\[
		|K_W\cap\cL_N|
		=
		\lambda_W^\star\binom{N+d}{d}
		+
		O_W(N^{d-1}),
		\]
		which proves \eqref{eq:lattice-count-asymp}.
		
		Finally, since
		\[
		\binom{N+d}{d}
		=
		\frac{N^d}{d!}\left(1+O\left(\frac1N\right)\right),
		\]
		we have
		\( 
		|\cU_N|
		=
		\lambda_W^\star\frac{N^d}{d!}
		+
		O_W(N^{d-1}).
		\)
		Taking logarithms completes the proof of \eqref{eq:log-cardinality}.
	\end{proof}
	
	\subsection{Coordinate Decoder}\label{sec:Coordinate Decoder}
	
	In this subsection, we introduce the nearest-neighbor decoder used in our
	achievability bound.
	
	Let \(A:=\aff(\cP_W)\) be the reachable output affine hull, and let
	\(\Pi_A:\bbR^m\to A\) denote the Euclidean projection onto \(A\). 
	We define the affine coordinate map
	\begin{equation}
		H(q):=T^{-1}(\Pi_A(q)),
		\qquad q\in\bbR^m .
		\label{eq:H-projection}
	\end{equation}
	
	Since \(\Pi_A\) is affine and the extended inverse
	\(T^{-1}:A\to\bbR^{d+1}\) is affine, \(H\) is an affine map from
	\(\bbR^m\) to \(\bbR^{d+1}\).
	Moreover, because
	\(\Pi_A (q)=q\) for all \(q\in A\), we have
	\begin{equation}
		H(T(u))=u,
		\qquad \forall u\in\Delta_d .\nonumber
	\end{equation}
	Thus \(H\) \footnote{After the minimizing simplex \( S_W^\star\) is fixed, the projection map \(H\) in \eqref{eq:H-projection} is fixed as well; this dependence is
		suppressed in constants such as \(O_W(\cdot)\) and \(O_{W,\epsilon}(\cdot)\).
	} first projects an empirical output distribution onto the reachable
	affine hull and then expresses it in the simplex coordinates induced by \(T\). 
	
	The projection in \eqref{eq:H-projection} leaves every reachable output
	distribution unchanged. For \(q_u\in\cP_W\), it removes only the component of
	the empirical fluctuation \(\widehat q_n-q_u\) orthogonal to the linear subspace
	parallel to \(A\), while all differences between reachable output laws lie in
	that subspace.

	Write \(L\) for the linear part of \(H\), so that
	\begin{equation}
		H(q)-H(q')=L(q-q'),
		\qquad q,q'\in\bbR^m .\nonumber
	\end{equation}

	For an observed output sequence \(Y^n\), define its empirical output distribution
	by
	\[
	\widehat P_{Y^n}(y)
	:=
	\frac1n\sum_{t=1}^n \mathbf 1\{Y_t=y\},
	\qquad y\in\cY.
	\]
	For brevity, we write \(\widehat q_n:=\widehat P_{Y^n}\).  Its coordinate image is
	\begin{equation}
		\widehat u_n:=H(\widehat q_n)
		=T^{-1}(\Pi_A\widehat q_n)
		\in\bbR^{d+1}.\nonumber
	\end{equation}
	Thus the decoder uses the empirical output distribution only through its
	projection onto the reachable affine hull, expressed in the affine coordinates
	induced by \(T\).

	The decoder chooses the unique nearest neighbor in the affine coordinates.  For
	\(z\in\bbR^{d+1}\), define the nearest-neighbor set
	\[
	\Gamma_N(z)
	:=
	\left\{
	v\in\cU_N:
	\|z-v\|_2=\min_{w\in\cU_N}\|z-w\|_2
	\right\}.
	\]
	Then define
	\begin{equation}
		g_N(z):=
		\begin{cases}
			u, & \text{if } \Gamma_N(z)=\{u\},\\
			\mathtt e, & \text{otherwise}.
		\end{cases}
		\label{eq:nearest-decoder}
	\end{equation}
	For the observed output sequence, the decoder output is
	\[
	\widehat u:=g_N(\widehat u_n),
	\]
	where ties are declared as decoding errors.

	\subsection{Error Reduction for Coordinate Decoder}\label{sec:Voronoi-Transfer Reduction}
	
	We next introduce an error reduction lemma tailored to the decoder introduced in Section~\ref{sec:Coordinate Decoder}.  It plays the role of the
	neighboring error reduction used in the full-dimensional ML analysis
	\cite[Lemma~1]{feng_lower_isit_2025}, but is formulated for Voronoi decision
	regions rather than likelihood-ratio comparisons.  Such
	decision regions are the usual Voronoi regions associated with Euclidean
	nearest-neighbor decoding and lattice quantization; see, e.g.,
	\cite{conway_sloane_sphere_1999,zamir_lattice_2014}.
	For a finite set $\cU_N\subseteq\Delta_d$ and a point
	$u\in\cU_N$, its decision region, or Voronoi
	region, is
	\begin{equation}
		\begin{aligned}
			\cV(u;\cU_N)
			:=
			\bigl\{z\in\bbR^{d+1}:\
			&\|z-u\|_2 < \|z-v\|_2,\\
			&\forall v\in\cU_N,\ v\ne u
			\bigr\}.
		\end{aligned}
		\label{eq:voronoi-region}
	\end{equation}
	We use strict inequalities because ties are counted as errors.
	Thus, if $u$ is transmitted, a decoding error occurs whenever
	$z$ lies outside $\cV(u;\cU_N)$, equivalently whenever
	there exists some $v\in\cU_N$, $v\ne u$, such that
	$\|z-v\|_2\le \|z-u\|_2$. For \(i\ne j\), the inner product
	\(\langle z-u,e_i-e_j\rangle=(z_i-u_i)-(z_j-u_j)\) measures the coordinate
	error along the elementary transfer direction \(e_i-e_j\).

	The following lemma gives an error reduction for the
	affine coordinate message set \(\cU_N\subseteq\cL_N\). 
	In the proof, we only use the fact that differences
	between two points of \( \cL_N\) are zero-sum
	integer vectors, and hence can be decomposed into elementary
	transfer directions $e_i-e_j$.
	
	\begin{lemma} 
		\label{lem:voronoi-transfer}
		Let \(\cL_N\) be the simplex lattice, and let
		\(\cU_N\subseteq\cL_N\) be an arbitrary finite message set.
		
		Fix \(u\in\cU_N\) and let \(z\in\bbR^{d+1}\) be an arbitrary point in the
		coordinate space. Consider the nearest-neighbor decoder \(g_N\) in
		\eqref{eq:nearest-decoder}. Suppose that \(z\) is not decoded uniquely as
		\(u\), i.e., \(g_N(z)\ne u\). Then there exists an elementary transfer
		direction \(e_i-e_j\), with \(0\le i,j\le d\) and \(i\ne j\), such that
		\begin{equation}
			\langle z-u,e_i-e_j\rangle\ge \frac1N .
			\label{eq:transfer-event}
		\end{equation}
		The number of ordered transfer directions is
		\begin{equation}
			R_d:=d(d+1).
			\label{eq:Rd}
		\end{equation}
	\end{lemma}
	\begin{remark}
		Lemma~\ref{lem:voronoi-transfer} replaces the neighboring-error reduction used
		in the full-dimensional ML analysis of \cite[Lemma~1]{feng_lower_isit_2025}.
		While both reductions localize the union over competing messages to finitely
		many elementary directions, the objects being localized are different.  The
		earlier reduction is formulated for likelihood-ratio comparisons between
		neighboring output distributions.  In the present construction, the message set, equivalently the set of target output laws, is represented as a simplex lattice in the affine coordinates \(u=T^{-1}(q)\), and
		the decoder is Euclidean nearest-neighbor decoding in these coordinates.  Hence
		the relevant error event is a Voronoi event rather than a likelihood-ratio
		event.
		The lemma shows that, although a Voronoi error may be caused by an arbitrary
		competing point \(v\in\cU_N\), the difference \(v-u\) is a zero-sum integer
		vector scaled by \(1/N\) and can be decomposed into elementary transfers
		\(e_i-e_j\).  Therefore every nearest-neighbor error crosses one of only
		\(d(d+1)\) transfer halfspaces.  This reduction allows the Gaussian analysis in
		Section~\ref{sec:refined-average-gaussian} to use one-dimensional
		affine coordinate fluctuations instead of likelihood ratios.
	\end{remark}
	
	\begin{proof}
		Since \(g_N(z)\ne u\), by the definition of \(g_N\) there exists
		\(v\in\cU_N\), \(v\ne u\), such that
		\begin{equation}
			\|z-v\|_2\le \|z-u\|_2 .
			\label{eq:nn-error-condition}
		\end{equation}
		Since \(\cU_N\subseteq\cL_N\), both \(u\) and \(v\) are scaled simplex
		lattice points. Hence there exists a nonzero vector
		\[
		k=(k_0,\ldots,k_d)\in\bbZ^{d+1},
		\qquad
		\sum_{i=0}^d k_i=0,
		\]
		such that
		\(
		v-u=\frac{k}{N}.
		\)
		Let
		\(
		a:=z-u .
		\)
		Then \eqref{eq:nn-error-condition} is equivalent to
		\(
		\left\|a-\frac{k}{N}\right\|_2^2\le \|a\|_2^2.
		\)
		Expanding the square gives
		\begin{equation}
			\langle a,k\rangle\ge \frac{\|k\|_2^2}{2N}.
			\label{eq:nearest-halfspace}
		\end{equation}
		
		Write
		\begin{align*}
			k&=k^+-k^-,\\
			k_i^+&=\max\{k_i,0\},\qquad
			k_i^-=\max\{-k_i,0\}.
		\end{align*}
		Since $\sum_i k_i=0$, the total positive and negative masses
		are equal. Define
		\[
		m(k):=\sum_i k_i^+=\sum_i k_i^-.
		\]
		The vector $k$ can be decomposed into $m(k)$ elementary
		transfer directions:
		\begin{equation}
			k=\sum_{\ell=1}^{m(k)}(e_{i_\ell}-e_{j_\ell}),
			\label{eq:k-transfer-decomposition}
		\end{equation}
		where the indices $i_\ell$ are chosen from the positive
		coordinates of $k$ and the indices $j_\ell$ are chosen from
		the negative coordinates of $k$, with multiplicities.
		
		Assume, for contradiction, that no elementary transfer
		direction satisfies \eqref{eq:transfer-event}. Then
		\[
		\langle a,e_i-e_j\rangle<\frac1N,
		\qquad
		\forall i\ne j.
		\]
		Using the decomposition \eqref{eq:k-transfer-decomposition},
		we obtain
		\begin{equation}
			\langle a,k\rangle
			=
			\sum_{\ell=1}^{m(k)}
			\langle a,e_{i_\ell}-e_{j_\ell}\rangle
			<
			\frac{m(k)}{N}.
			\label{eq:ak-upper}
		\end{equation}
		
		On the other hand, since $k_i$ is integer-valued, we have
		$k_i^2\ge k_i^+$ when $k_i>0$ and
		$k_i^2\ge k_i^-$ when $k_i<0$. Therefore,
		\[
		\|k\|_2^2
		=
		\sum_i k_i^2
		\ge
		\sum_i k_i^+ + \sum_i k_i^-
		=
		2m(k).
		\]
		Thus,
		\begin{equation}
			\frac{m(k)}{N}
			\le
			\frac{\|k\|_2^2}{2N}.
			\label{eq:mass-norm-comparison}
		\end{equation}
		Combining \eqref{eq:nearest-halfspace},
		\eqref{eq:ak-upper}, and \eqref{eq:mass-norm-comparison}
		yields the contradiction
		\[
		\frac{m(k)}{N}
		>
		\langle a,k\rangle
		\ge
		\frac{\|k\|_2^2}{2N}
		\ge
		\frac{m(k)}{N}.
		\]
		Therefore, at least one ordered pair $(i,j)$ with $i\ne j$
		must satisfy
		\( 
		\langle a,e_i-e_j\rangle\ge\frac1N.
		\)
		Since $a=z-u$, this proves \eqref{eq:transfer-event}. The number of ordered
		transfer directions is the number of ordered pairs \((i,j)\) with
		\(0\le i,j\le d\) and \(i\ne j\), namely \((d+1)d\).  This proves
		\eqref{eq:Rd}.
	\end{proof}
	 	
	\subsection{Finite-Blocklength Achievability}\label{sec:Low-Rank Gaussian Achievability}
	
	We briefly clarify the coordinate form of the decoder.  A message
	\(u\in\cU_N\subseteq\Delta_d\) is associated with the reachable output law
	\(q_u=T(u)\in\cP_W\).  Given \(Y^n\), the decoder forms the empirical output
	distribution \(\widehat q_n\), projects it onto
	\(A=\aff(\cP_W)\), and maps it back to the coordinate space:
	\[
	\widehat u_n
	=
	H(\widehat q_n)
	=
	T^{-1}(\Pi_A\widehat q_n).
	\]
	Since \(T(u)\in A\), the definition of \(H\) gives \(H(T(u))=u\). Hence this
	coordinate representation is consistent with the message parametrization.  The decoder then chooses the Euclidean nearest
	neighbor of \(\widehat u_n\) in \(\cU_N\), with ties declared as errors.
	
	We now give our main result in this section.
	\begin{theorem}
		\label{thm:finite-achievability}
		For every \(n\ge1\) and every \(N\ge1\) such that
		\(\cU_N\ne\varnothing\), there exists a code $\cC_n$ with message set
		\(\cU_N\) whose average error probability satisfies
		\begin{align}
			P_e 
			&\le
			\frac1{|\cU_N|}
			\sum_{u\in\cU_N}
			\sum_{\substack{0\le i,j\le d\\i\ne j}}
			\bbP_u\left[
			\langle\widehat u_n-u,e_i-e_j\rangle\ge\frac1N
			\right],
			\label{eq:finite-achievability-union}
		\end{align}
		where \(\bbP_u\) denotes the law induced by \(X^n\sim p_u^{\otimes n}\), the DMC,
		and the random permutation block.
	\end{theorem}
	\begin{proof}
		For each \(u\in\cU_N\), since \(T(u)\in\cP_W\), choose
		\(p_u\in\Delta_{k-1}\) satisfying
		\[
		p_uW=T(u)=:q_u .
		\]
		
		Generate a random codebook as follows.  For each \(u\in\cU_N\), draw the
		codeword \(X^n(u)\) independently according to \(p_u^{\otimes n}\). Consider
		the error probability conditioned on message \(u\), averaged over this random
		codebook and over the channel.  After averaging over the random choice of
		\(X^n(u)\), the DMC output satisfies \(Z^n\sim q_u^{\otimes n}\).  The permutation block is uniform and independent; because
		\(q_u^{\otimes n}\) is permutation-invariant, the output also satisfies
		\(Y^n\sim q_u^{\otimes n}\).  Thus, under the averaged law \(\bbP_u\),
		\(Y_1,\ldots,Y_n\) are i.i.d. according to \(q_u\)
		\cite{makur_coding_2020}.
		
		The nearest-neighbor decoder computes \(\widehat q_n=\widehat P_{Y^n}\) and
		\(\widehat u_n=H(\widehat q_n)\), and then applies \(g_N\).  Under message
		\(u\), the error event is \(\{g_N(\widehat u_n)\ne u\}\).  Applying
		Lemma~\ref{lem:voronoi-transfer} to each realization with \(z=\widehat u_n\)
		gives
		\[
		\{g_N(\widehat u_n)\ne u\}
		\subseteq
		\bigcup_{\substack{0\le i,j\le d\\i\ne j}}
		\left\{
		\langle\widehat u_n-u,e_i-e_j\rangle\ge\frac1N
		\right\}.
		\]
		Averaging over the uniformly distributed message on \(\cU_N\), over the random codebook,
		and over the channel, and then applying the union bound, we obtain
		\[
		\mathbb E_{\mathcal C} P_e(\mathcal C)
		\le
		\frac1{|\cU_N|}
		\sum_{u\in\cU_N}
		\sum_{\substack{0\le i,j\le d\\i\ne j}}
		\bbP_u\left[
		\langle\widehat u_n-u,e_i-e_j\rangle\ge\frac1N
		\right],
		\]
		where \(\mathbb E_{\mathcal C}\) denotes expectation with respect to the
		random codebook ensemble.  Therefore, there exists a code \(\cC_n\), whose average error probability is no larger than
		the right-hand side above.
	\end{proof}

	\section{Converse Bounds}
	\label{sec:converse-covering}
	
	In this section, we prove a fixed-error converse bound. The proof follows the
	meta-converse and \textit{information-spectrum} approaches in
	\cite{polyanskiy_channel_2010,wang_simple_2009,polyanskiy_saddle_2013,tomamichel_tight_2013}:
	the channel-coding problem is reduced to binary hypothesis testing, and the
	auxiliary output distribution is constructed from a divergence covering of the
	reachable output polytope \(\cP_W\). The key point is that, in the
	fixed-error setting, a lower-dimensional covering alone is not sufficient:
	one must also control the binary test between the true noisy permutation
	transition kernel and the nearby product auxiliary distributions induced by
	the covering. We therefore retain the noisy permutation transition kernel in
	the hypothesis test defining \(\beta_\alpha\) and prove a uniform local
	hypothesis-testing bound. This bound shows that, when the covering radius is
	chosen as \(a^2/n\), the associated local testing term remains bounded
	uniformly in \(n\). Consequently, the \(n\)-dependent part of the converse is
	governed by the divergence covering number of \(\cP_W\), while the local testing and meta-converse terms are absorbed into \(O_{W,\epsilon,a}(1)\). Divergence covering and related minimax redundancy
	ideas are closely connected to
	\cite{yang_information-theoretic_1999,Jennifer_2022}.

	\subsection{Divergence Covering}
	For \(\mathcal A\subseteq\Delta_{m-1}\) and \(\rho>0\), define the KL-covering number
	\[
	N_D(\mathcal A,\rho)
	:=
	\min\left\{
	|\mathcal G|:\mathcal G\subseteq\mathcal A,\ 
	\sup_{p\in\mathcal A}\inf_{q\in\mathcal G}D(p\|q)\le \rho
	\right\}.
	\]
	
	\begin{lemma}[Subspace divergence covering]
		\label{lem:subspace-div-covering}
		Under Assumption~\ref{ass:positive-support}, let
		\(d=\dim\cP_W\ge1\).  There exist constants \(C_W<\infty\) and
		\(r_0>0\) such that, for every \(0<r<r_0\),
		\begin{equation}
			\log N_D(\cP_W,r)
			\le
			d\log\frac{1}{\sqrt r}
			+
			C_W .
			\label{eq:subspace-covering-general}
		\end{equation}
	\end{lemma}
	
	\begin{proof}
		The reachable output polytope can be written as
		\[
		\cP_W=\{pW:p\in\Delta_{|\cX|-1}\}.
		\]
		Thus \(\cP_W\) is the image of an input simplex under the stochastic matrix
		\(W\).  Since \(\dim\cP_W=d\), the rank of this affine image is \(d+1\).
		By the subspace covering argument of
		\cite[Proposition~2]{tang_capacity_2023}, the KL covering number of
		\(\cP_W\) is bounded by the KL covering number of a \(d\)-dimensional
		simplex, up to a constant depending only on \(W\).  Combining this with the
		simplex KL-covering bound in
		\cite[Theorem~4]{tang_capacity_2023} gives
		\[
		N_D(\cP_W,r)
		\le
		C'_W r^{-d/2}
		\]
		for all sufficiently small \(r\).  Taking logarithms proves
		\eqref{eq:subspace-covering-general}.
	\end{proof}
	
	\subsection{Local testing and Meta-Converse}
	\label{sec:constant-remainder-converse}

	In this subsection, we prove a local testing bound and a
	meta-converse tailored to the noisy permutation channel.  The testing estimate
	bounds the type-II error between the noisy permutation transition kernel
	\(P_{Y^n|x^n}^{\rm perm}\) and a nearby product output distribution
	\(\bar q^{\otimes n}\).  We first introduce the notation needed for this bound.
	
	For a deterministic input sequence \(x^n\in\mathcal X^n\), let
	\(\pi_{x^n}\in\Delta_{|\mathcal X|-1}\) denote its input type:
	\[
	\pi_{x^n}(x)
	:=
	\frac1n\sum_{t=1}^n \mathbf 1\{x_t=x\},
	\qquad x\in\mathcal X .
	\]
	Define the corresponding average output distribution
	\[
	q_{x^n}
	:=
	\pi_{x^n}W
	=
	\frac1n\sum_{t=1}^n W(\cdot|x_t)
	\in\mathcal P_W .
	\]
	
	Let \(P_{Y^n|x^n}^{\rm perm}\) denote the output distribution of the
	noisy permutation channel when the deterministic input sequence is
	\(x^n\). Equivalently, if \(Z_1,\ldots,Z_n\) are conditionally
	independent given \(x^n\), with
	\[
	Z_t\sim W(\cdot|x_t),
	\]
	and \(\sigma\) is an independent uniformly random permutation of
	\(\{1,\ldots,n\}\), then
	\[
	Y_{\sigma(t)}=Z_t, \qquad t=1,\ldots,n,
	\]
	and \(P_{Y^n|x^n}^{\rm perm}\) is the law of \(Y^n\).
	
	Recall that \(N(y^n)\) denotes the output count vector of \(y^n\).
	For \(x^n\in\mathcal X^n\), define
	\[
	q_{x^n}:=\pi_{x^n}W,
	\qquad
	P_{x^n}:=P_{Y^n|x^n}^{\rm perm}.
	\]
	
	We will use the following lemmas.
	\begin{lemma} 
		\label{lem:count-space-likelihood-ratio}
		Let \(x^n\in\cX^n\) be deterministic.  Let \(P_{x^n}^{\rm cnt}\) denote the law
		of \(N(Y^n)\) under \(P_{x^n}\).  For any distribution \(r\) on \(\cY\) with
		\(r(y)>0\) for all \(y\), and every \(y^n\in\cY^n\) with
		\(N(y^n)=\nu\),
		\begin{equation}
			\log\frac{dP_{x^n}}{dr^{\otimes n}}(y^n)
			=
			\log
			\frac{
				P_{x^n}^{\rm cnt}(\nu)
			}{
				\operatorname{Mult}(n,r)(\nu)
			}.
			\label{eq:count-space-likelihood-ratio}
		\end{equation}
	\end{lemma}
	
	\begin{proof}
		Both \(P_{x^n}\) and \(r^{\otimes n}\) are exchangeable, and hence both are
		constant on each output type class
		\[
		\mathcal T_\nu:=\{y^n:N(y^n)=\nu\}.
		\]
		For \(y^n\in\mathcal T_\nu\),
		\[
		P_{x^n}(y^n)=\frac{P_{x^n}^{\rm cnt}(\nu)}{|\mathcal T_\nu|},
		\qquad
		r^{\otimes n}(y^n)
		=
		\frac{\operatorname{Mult}(n,r)(\nu)}{|\mathcal T_\nu|}.
		\]
		Taking the ratio proves the claim.
	\end{proof}

	\begin{lemma}
		\label{lem:uniform-count-local}
		Assume \(p_{\min}>0\). Let \(P_1,\ldots,P_n\) be distributions on
		\(\mathcal Y=\{1,\ldots,m\}\) with \(P_s(y)\ge p_{\min}\) for all \(s\) and
		all \(y\). Let \(b_y\) denote the \(y\)-th standard basis vector in
		\(\mathbb R^m\), and let \(B_s\in\{b_1,\ldots,b_m\}\) be independent random
		vectors with
		\[
		\mathbb P[B_s=b_y]=P_s(y).
		\]
		Set
		\[
		S_n:=\sum_{s=1}^n B_s,\qquad
		\bar P_n:=\frac1n\sum_{s=1}^n P_s .
		\]
		Then the following uniform estimates hold.
		
		\begin{enumerate}
			\item[(i)]
			For every \(\alpha\in(0,1)\), there exists
			\(K_\alpha<\infty\), depending only on
			\((p_{\min},m,\alpha)\), such that
			\[
			\mathbb P\left[\|S_n-n\bar P_n\|_2\le K_\alpha\sqrt n\right]
			\ge 1-\frac{\alpha}{2}.
			\]
			
			\item[(ii)]
			There exists \(c_2<\infty\), depending only on
			\((p_{\min},m)\), such that
			\[
			\sup_{\nu}\mathbb P[S_n=\nu]
			\le c_2 n^{-(m-1)/2},
			\]
			where the supremum is over all count vectors
			\(\nu\in\mathbb Z_{\ge0}^m\) satisfying
			\(\sum_{y=1}^m\nu_y=n\).
			
			\item[(iii)]
			For every \(C_0<\infty\), there exist constants
			\(c_1>0\) and \(n_0<\infty\), depending only on
			\((p_{\min},m,C_0)\), such that for all \(n\ge n_0\), all
			\(q\in\Delta_{m-1}\) with \(q(y)\ge p_{\min}\), and all count vectors
			\(\nu\in\mathbb Z_{\ge0}^m\) satisfying
			\[
			\sum_{y=1}^m\nu_y=n,
			\qquad
			\|\nu-nq\|_2\le C_0\sqrt n,
			\]
			we have
			\[
			\operatorname{Mult}(n,q)(\nu)
			\ge c_1 n^{-(m-1)/2}.
			\]
		\end{enumerate}
	\end{lemma}

	\begin{proof}
		See Appendix \ref{apx:proof_uniform-count-local}.
	\end{proof}

	The following lemma is the key local estimate underlying the bounded-remainder
	converse. Its main idea is that if \(\bar q\in\mathcal P_W\) satisfies
	\(
	D(q_{x^n}\|\bar q)=O\!\left(\frac1n\right),
	\)
	then Pinsker's inequality implies that the two mean count vectors
	\(nq_{x^n}\) and \(n\bar q\) are separated by only \(O(\sqrt n)\). Consequently,
	a local count region around \(nq_{x^n}\) is contained, after enlarging the
	constant, in a local count region around \(n\bar q\). This allows us to compare
	the noisy permutation law \(P_{x^n}\) with the product law
	\(\bar q^{\otimes n}\) on the same \(O(\sqrt n)\)-scale count region. On this
	region, the count probabilities under the two laws are of the same polynomial
	order, while the likelihood ratio remains uniformly bounded. This yields a
	constant lower bound on \(\beta_\alpha\), uniformly over \(x^n\), \(\bar q\),
	and \(n\).

	\begin{lemma}
		\label{lem:symbol-relaxed-intrinsic}
		Under Assumption~\ref{ass:positive-support}, fix \(\alpha\in(0,1)\) and \(a>0\). There exist constants \(C_{W,\alpha,a}<\infty\) and \(n_0<\infty\) such that, for all \(n\ge n_0\), every \(x^n\in\mathcal X^n\), and every \(\bar q\in\mathcal P_W\) satisfying
		\[
		D(q_{x^n}\|\bar q)\le \frac{a^2}{n},
		\]
		we have
		\[
		-\log \beta_\alpha
		\!\left(
		P_{x^n},
		\bar q^{\otimes n}
		\right)
		\le
		C_{W,\alpha,a}.
		\]
	\end{lemma}

	\begin{proof}
		Set 
		\[ 
		q:=q_{x^n},\qquad P:=P_{x^n},\qquad Q:=\bar q^{\otimes n}. 
		\]
		Under \(P\), the count vector \(N(Y^n)\) has the same law as
		\(
		S_P:=\sum_{s=1}^n B_s,
		\)
		where \(B_1,\ldots,B_n\) are independent random vectors taking values in
		\(\{b_1,\ldots,b_m\}\) and satisfying
		\(
		\mathbb P[B_s=b_y]=W(y|x_s)\) for \(y\in\cY .
		\)
		Hence \(\mathbb E[S_P]=nq\).  Under \(Q\), the count vector \(N(Y^n)\) has the
		multinomial law with parameters \((n,\bar q)\); denote this count vector by
		\(\bar S\).

		By Lemma~\ref{lem:count-space-likelihood-ratio}, for every
		\(y^n\) with \(N(y^n)=t\),
		\begin{equation}
		\frac{P(y^n)}{Q(y^n)}
		=
		\frac{\mathbb P[S_P=t]}{\mathbb P[\bar S=t]} . \label{eq:count-likelihood-ratio}
		\end{equation}

		By Pinsker's inequality and \(D(q\|\bar q)\le a^2/n\),
		\[
		\|q-\bar q\|_2
		\le
		C_a n^{-1/2}
		\]
		for a constant \(C_a<\infty\) depending only on \(a\) and the logarithm
		convention.  Apply Lemma~\ref{lem:uniform-count-local} with
		\(P_s=W(\cdot|x_s)\), \(s=1,\ldots,n\). Uniformly over \(x^n\), there exist
		\(K_\alpha<\infty\) and \(c_2<\infty\) such that
		\[
		\begin{aligned}
			&P\!\left[\|N(Y^n)-nq\|_2\le K_\alpha\sqrt n\right]
			\nonumber \\
			& =
			\mathbb P\!\left[\|S_P-nq\|_2\le K_\alpha\sqrt n\right]\\
			&\ge 1-\frac{\alpha}{2}
		\end{aligned}
		\]
		and
		\[
		\sup_t \mathbb P[S_P=t]\le c_2 n^{-(m-1)/2}.
		\]
		
		Moreover, if \(\|t-nq\|_2\le K_\alpha\sqrt n\), then
		\[
		\|t-n\bar q\|_2
		\le
		(K_\alpha+C_a)\sqrt n .
		\]
		Since \(\bar q\in\mathcal P_W\) and all distributions in \(\mathcal P_W\) are
		bounded below by \(p_{\min}\), the multinomial lower bound in
		Lemma~\ref{lem:uniform-count-local} gives constants \(c_1>0\) and \(n_1\),
		depending only on \((W,\alpha,a)\), such that for all \(n\ge n_1\),
		\[
		\mathbb P[\bar S=t]\ge c_1 n^{-(m-1)/2}
		\]
		for every \(t\) satisfying \(\|t-nq\|_2\le K_\alpha\sqrt n\).
		
		Define
		\[
		\mathcal A
		:=
		\{y^n:\|N(y^n)-nq\|_2\le K_\alpha\sqrt n\}.
		\]
		Then \(P(\mathcal A)\ge1-\alpha/2\), and by
		\eqref{eq:count-likelihood-ratio},
		\[
		\frac{P(y^n)}{Q(y^n)}
		\le
		\frac{c_2}{c_1}
		=:L_{W,\alpha,a},
		\qquad y^n\in\mathcal A .
		\]
		Now let \(\varphi\) be any randomized test with \(P[\varphi=1]\ge\alpha\).
		Since \(P(\mathcal A^c)\le\alpha/2\),
		\[
		P[\varphi=1,\mathcal A]\ge \frac{\alpha}{2}.
		\]
		Using \(dP/dQ\le L_{W,\alpha,a}\) on \(\mathcal A\), we get
		\[
		Q[\varphi=1]
		\ge
		Q[\varphi=1,\mathcal A]
		\ge
		L_{W,\alpha,a}^{-1}P[\varphi=1,\mathcal A]
		\ge
		\frac{\alpha}{2L_{W,\alpha,a}} .
		\]
		Taking the infimum over all such tests yields
		\[
		\beta_\alpha(P,Q)
		\ge
		\frac{\alpha}{2L_{W,\alpha,a}},
		\]
		and therefore
		\[
		-\log\beta_\alpha
		\left(
		P_{Y^n|x^n}^{\rm perm},
		\bar q^{\otimes n}
		\right)
		\le
		\log\left(\frac{2L_{W,\alpha,a}}{\alpha}\right)
		=:C_{W,\alpha,a}.
		\]
	\end{proof}
	
	\begin{remark}

		The covering radius \(1/n\) in Lemma~\ref{lem:symbol-relaxed-intrinsic}
		is the natural local scale for the divergence-covering converse.  The
		lemma shows that, when the covering point satisfies
		\(D(q\|\bar q)=O(1/n)\), the testing penalty remains uniformly bounded
		in \(n\).  This boundedness conclusion cannot in general be sharpened
		to an \(o(1)\) testing penalty: under such \(O(1/n)\) perturbations, the
		testing term may differ by a non-vanishing constant from the reference
		value \(-\log\alpha\) attained when the two hypotheses coincide.

		To see this, consider a binary symmetric channel with crossover probability
		\(\delta\in(0,1/2)\), and let the input sequence be \(x^n=(0,\ldots,0)\).
		Then the noisy permutation output law is
		\[
		P_n=q^{\otimes n},
		\qquad
		q=(1-\delta,\delta),
		\]
		since the permutation block has no effect on an i.i.d. output sequence.  Fix
		\(s>0\) and define
		\[
		\bar q_n
		=
		\left(
		1-\delta-\frac{s}{\sqrt n},
		\delta+\frac{s}{\sqrt n}
		\right).
		\]
		For all sufficiently large \(n\), \(\bar q_n\) lies in the reachable output
		interval of the BSC.  Moreover, with base-two logarithms,
		\[
		D(q\|\bar q_n)
		=
		\frac{1}{2\ln 2}\,
		\frac{s^2}{\delta(1-\delta)}\,
		\frac1n
		+
		o\!\left(\frac1n\right).
		\]
		Thus, by choosing \(s\) sufficiently small, the condition
		\[
		D(q\|\bar q_n)\le \frac{a^2}{n}
		\]
		is satisfied for all sufficiently large \(n\).
		
		Let \(Q_n=\bar q_n^{\otimes n}\), and define the natural-log likelihood ratio
		\[
		\ell_n:=\ln\frac{dP_n}{dQ_n}.
		\]
		With
		\[
		J:=\frac{s^2}{\delta(1-\delta)},
		\]
		the standard local asymptotic expansion gives
		\[
		\ell_n \Rightarrow \mathcal N\left(\frac J2,J\right)
		\quad \text{under } P_n
		\]
		and
		\[
		\ell_n \Rightarrow \mathcal N\left(-\frac J2,J\right)
		\quad \text{under } Q_n .
		\]
		Therefore, by the Neyman-Pearson lemma and the likelihood-ratio
		test,
		\[
		\lim_{n\to\infty}
		\beta_\alpha(P_n,Q_n)
		=
		\Phi\bigl(\Phi^{-1}(\alpha)-\sqrt J\bigr).
		\]
		For every fixed \(s>0\), this limit is strictly smaller than \(\alpha\).
		Hence
		\[
		\begin{aligned}
			&\lim_{n\to\infty}
			\left[
			-\log \beta_\alpha(P_n,Q_n)-(-\log\alpha)
			\right] \\
			&\quad =
			-\log
			\frac{
				\Phi\bigl(\Phi^{-1}(\alpha)-\sqrt J\bigr)
			}{\alpha}
			>0 .
		\end{aligned}
		\]
		Thus, even when the single-letter KL distance is of order \(1/n\), the
		binary-testing term can contain a non-vanishing constant contribution.  This
		is the sense in which Lemma~\ref{lem:symbol-relaxed-intrinsic} is order-sharp
		for the bounded-remainder converse.
	\end{remark}
	
	We next state the meta-converse tailored to the noisy permutation channel
	\begin{lemma} 
		\label{lem:symbol-relaxed-meta-converse-perm}
		Let
		\[
		V_n(y^n|x^n):=P_{Y^n|x^n}^{\rm perm}(y^n)
		\]
		be the \(n\)-block transition kernel of the noisy permutation channel. Fix \(\epsilon\in(0,1)\) and \(\eta\in(\epsilon,1)\).  Then, for every
		auxiliary output distribution \(Q_{Y^n}\) on \(\mathcal Y^n\), every code
		\(\cC_n\) with message set \(\cM\) and average error probability at most
		\(\epsilon\) satisfies
		\[
		\begin{aligned}
			&\log |\cM| \nonumber \\
			&\le
			\sup_{x^n\in\mathcal X^n}
			\left\{
			-\log
			\beta_{1-\eta}
			\left(
			V_n(\cdot|x^n),
			Q_{Y^n}
			\right)
			\right\}
			-
			\log\left(1-\frac{\epsilon}{\eta}\right).
		\end{aligned}
		\]
	\end{lemma}
	
	\begin{proof}
		Let the messages be equiprobable, and let \(D_m\subseteq\mathcal Y^n\)
		be the decoding region of message \(m\).  Denote the conditional error
		probability of message \(m\) by
		\[
		e_m
		:=
		V_n(D_m^c|x_m^n).
		\]
		The average error assumption gives
		\(
		\frac{1}{|\cM|}\sum_{m=1}^{|\cM|} e_m\le \epsilon .
		\)
		Fix \(\eta\in(\epsilon,1)\), and define the set of good messages
		\[
		\mathcal G:=\{m:e_m\le \eta\}.
		\]
		By Markov's inequality,
		\[
		|\mathcal G|
		\ge
		|\cM|\left(1-\frac{\epsilon}{\eta}\right).
		\]
		
		For each \(m\in\mathcal G\), the test \(T_m=\mathbf 1\{Y^n\in D_m\}\)
		has power at least \(1-\eta\) under \(V_n(\cdot|x_m^n)\).  Hence
		\[
		\beta_{1-\eta}
		\left(
		V_n(\cdot|x_m^n),
		Q_{Y^n}
		\right)
		\le
		Q_{Y^n}(D_m).
		\]
		Let
		\[
		A:=
		\sup_{x^n\in\mathcal X^n}
		\left\{
		-\log
		\beta_{1-\eta}
		\left(
		V_n(\cdot|x^n),
		Q_{Y^n}
		\right)
		\right\}.
		\]
		Then every \(m\in\mathcal G\) satisfies
		\[
		\beta_{1-\eta}
		\left(
		V_n(\cdot|x_m^n),
		Q_{Y^n}
		\right)
		\ge 2^{-A}.
		\]
		Since the decoding regions are disjoint,
		\[
		|\mathcal G|2^{-A}
		\le
		\sum_{m\in\mathcal G}\beta_{1-\eta}
		\left(
		V_n(\cdot|x_m^n),
		Q_{Y^n}
		\right)
		\le
		\sum_{m\in\mathcal G}Q_{Y^n}(D_m)
		\le 1.
		\]
		Thus \(\log|\mathcal G|\le A\).  Combining this with the lower bound
		on \(|\mathcal G|\) gives
		\[
		\log |\cM|
		\le
		A
		-
		\log\left(1-\frac{\epsilon}{\eta}\right),
		\]
		which proves the claim.
	\end{proof}
	
	\subsection{Covering Converse}
	In this subsection, we present a converse bound via divergence covering.
	We keep the covering radius \(a^2/n\) explicit because it is the natural
	local scale at which the covering number contributes
	\(d\log\sqrt n\) and the local testing term remains bounded.
	\begin{theorem}
		\label{thm:low-rank-converse-message-level}
		Under Assumption~\ref{ass:positive-support}, let
		\(d=\dim\mathcal P_W\ge1\).  For every
		\(\epsilon\in(0,1)\), there exists \(C_{W,\epsilon}<\infty\) such
		that, for all sufficiently large \(n\),
		\[
		\log M^\star(n,\epsilon)
		\le
		\frac d2\log n
		+
		C_{W,\epsilon}.
		\]
	\end{theorem}
	
	\begin{proof}
		Set \(\eta=(1+\epsilon)/2\) and
		\(\alpha=1-\eta=(1-\epsilon)/2\).  Choose the covering radius
		\(\rho_n=1/n\).  By Lemma~\ref{lem:subspace-div-covering}, there exist
		constants \(B_W<\infty\) and \(n_1<\infty\) such that, for all
		\(n\ge n_1\), there exists a KL covering
		\(\mathcal G_n\subseteq\mathcal P_W\) with radius \(\rho_n\) satisfying
		\[
		\log|\mathcal G_n|
		\le
		\frac d2\log n+B_W .
		\]
		Define
		\[
		Q_{Y^n}^{(\mathcal G_n)}
		:=
		\frac1{|\mathcal G_n|}
		\sum_{\bar q\in\mathcal G_n}\bar q^{\otimes n}.
		\]
		
		For any deterministic input sequence \(x^n\), let
		\(q_{x^n}:=\pi_{x^n}W\in\mathcal P_W\), and choose
		\(\bar q(x^n)\in\mathcal G_n\) satisfying
		\[
		D(q_{x^n}\|\bar q(x^n))\le \frac1n .
		\]
		Since
		\[
		Q_{Y^n}^{(\mathcal G_n)}
		\ge
		|\mathcal G_n|^{-1}\bar q(x^n)^{\otimes n},
		\]
		we have
		\[
		\begin{aligned}
			&-\log\beta_\alpha
			\left(
			P_{Y^n|x^n}^{\rm perm},
			Q_{Y^n}^{(\mathcal G_n)}
			\right)  
			\le
			\log|\mathcal G_n|
			-
			\log\beta_\alpha
			\left(
			P_{Y^n|x^n}^{\rm perm},
			\bar q(x^n)^{\otimes n}
			\right).
		\end{aligned}
		\]
		By Lemma~\ref{lem:symbol-relaxed-intrinsic} with \(a=1\), the second term is
		bounded by a constant \(C_{W,\alpha,1}\), uniformly in \(x^n\).  Therefore
		\[
		\sup_{x^n}
		-\log\beta_\alpha
		\left(
		P_{Y^n|x^n}^{\rm perm},
		Q_{Y^n}^{(\mathcal G_n)}
		\right)
		\le
		\frac d2\log n+B_W+C_{W,\alpha,1}.
		\]
		Applying Lemma~\ref{lem:symbol-relaxed-meta-converse-perm} gives
		\[
		\log M^\star(n,\epsilon)
		\le
		\frac d2\log n
		+
		B_W+C_{W,\alpha,1}
		-
		\log\left(1-\frac{\epsilon}{\eta}\right).
		\]
		The remaining terms depend only on \(W\) and \(\epsilon\), so they are absorbed
		into \(C_{W,\epsilon}\).
	\end{proof}
	
	Finally, we have the following strong converse result for logarithmic \(\epsilon\)-capacity.
	\begin{corollary} 
		Under Assumption~\ref{ass:positive-support}, let
		\(d=\dim\mathcal P_W\ge1\). For every
		\(\epsilon\in(0,1)\),
		\[
		C_\epsilon=\frac d2 .
		\]
	\end{corollary}
	\begin{proof}
		The converse follows from
		Theorem~\ref{thm:low-rank-converse-message-level}.  For achievability,
		the achievability result of \cite{makur_coding_2020} gives a sequence of
		noisy permutation channel codes with error probability tending to zero and
		logarithmic rate approaching
		\[
		\frac{\operatorname{rank}(W)-1}{2}
		=
		\frac{\dim\mathcal P_W}{2}.
		\]
		Hence, for every fixed \(\epsilon\in(0,1)\), these codes are admissible for
		\(M^\star(n,\epsilon)\) for all sufficiently large \(n\).  This gives
		\(C_\epsilon\ge d/2\).  Combining this with the converse bound gives
		\(C_\epsilon=d/2\).
	\end{proof}

	\section{Gaussian Approximation of Achievability}
	\label{sec:refined-average-gaussian}

	In this section, we present the achievability analysis by keeping the local Gaussian variance of each elementary transfer direction. This yields a Gaussian approximation for the simplex lattice construction.
	We then compare it with the bounded-remainder converse, whose leading
	\(n\)-dependent term is governed by the same affine dimension \(d\).
	
	Let
	\[
	\mathcal I_d
	:=
	\{(i,j):0\le i,j\le d,\ i\ne j\}
	\]
	be the set of ordered elementary transfer directions.  For
	\(u\in K_W\), let
	\[
	q_u:=T(u)\in \mathcal P_W
	\]
	be the corresponding output distribution. Recall that \(b_y\) is
	the standard basis vector associated with output symbol
	\(y\in\mathcal Y\). Let \(H\) be the affine map defined in
	\eqref{eq:H-projection}.
	Thus, if
	\(Y\sim q_u\), then
	\[
	\mathbb E_{q_u}[H(b_Y)]=u.
	\]
	For each ordered pair \((i,j)\in\mathcal I_d\), define the
	one-dimensional transfer fluctuation
	\[
	Z_{ij,u}(Y)
	:=
	\bigl\langle H(b_Y)-u,e_i-e_j\bigr\rangle ,
	\]
	and its variance
	\begin{equation}
		V_{ij}(u)
		:=
		\operatorname{Var}_{Y\sim q_u}
		\bigl[
		Z_{ij,u}(Y)
		\bigr].
		\label{eq:Vij-definition}
	\end{equation}
	We have the following useful lemma.
	\begin{lemma} 
		\label{lem:uniform-moments}
		Under Assumption~\ref{ass:positive-support}, let \(d=\dim\mathcal P_W\ge1\). There exist constants
		\(0<V_{\min}\le V_{\max}<\infty\) and \(\tau_{\max}<\infty\), depending
		only on \(W\) and on the chosen affine coordinates, such that, for all
		\(u\in K_W\) and all \((i,j)\in\mathcal I_d\),
		\begin{equation}
			V_{\min}
			\le
			V_{ij}(u)
			\le
			V_{\max},
			\qquad
			\mathbb E_{q_u}|Z_{ij,u}(Y)|^3\le \tau_{\max}.
			\label{eq:uniform-moment-bounds}
		\end{equation}
	\end{lemma}
	
	\begin{proof}
		See Appendix~\ref{apx_proof_uniform_moment}.
	\end{proof}
	
	For \(c>0\), define the average Gaussian union function
	\begin{equation}
		A_W(c)
		:=
		\frac{1}{\operatorname{vol}_d(K_W)}
		\int_{K_W}
		\sum_{(i,j)\in\mathcal I_d}
		\Phi\left(
		-\frac{1}{c\sqrt{V_{ij}(u)}}
		\right)
		\,du .
		\label{eq:average-control-functional}
	\end{equation}
	The function \(A_W(c)\) is continuous and nondecreasing in
	\(c\).  For a target average error probability \(\epsilon\), define the average Gaussian coefficient by
	\begin{equation}
		c_\epsilon
		:=
		\sup\{c>0:A_W(c)<\epsilon\}.
		\label{eq:c-epsilon-average-control}
	\end{equation}
	
	Since \(V_{ij}(u)\) is uniformly bounded above and below by
	Lemma~\ref{lem:uniform-moments}, dominated convergence gives
	\[
	\lim_{c\to 0^+}A_W(c)=0,
	\qquad
	\lim_{c\to\infty}A_W(c)=\frac{R_d}{2}.
	\]
	Because \(R_d=d(d+1)\ge2\) for \(d\ge1\), we have
	\(R_d/2\ge1>\epsilon\).  Hence, for every
	\(\epsilon\in(0,1)\), the coefficient \(c_\epsilon\) defined in
	\eqref{eq:c-epsilon-average-control} satisfies
	\(0<c_\epsilon<\infty\).

	We then have the following result.
	
	\begin{proposition}
		\label{prop:refined-average-achievability}
		Under Assumption~\ref{ass:positive-support}, let
		\(d=\dim\mathcal P_W\ge1\).  Fix \(\epsilon\in(0,1)\), and let
		\(c_\epsilon\) be defined in \eqref{eq:c-epsilon-average-control}.
		For every \(c\in(0,c_\epsilon)\),
		\[
		\log M^\star(n,\epsilon)
		\ge
		\frac d2\log n+d\log c+\log\lambda_W^\star-\log d!+o(1).
		\]
	\end{proposition}
	
	\begin{proof}
		Fix \(c<c_\epsilon\) and set \(N_n=\lfloor c\sqrt n\rfloor\).
		Since \(N_n\to\infty\), Lemma~\ref{lem:lattice-counting} implies that
		\(\mathcal U_{N_n}\ne\varnothing\) for all sufficiently large \(n\).  By
		Theorem~\ref{thm:finite-achievability}, there exists a code with message set
		\(\mathcal U_{N_n}\) whose average error probability satisfies
		\[
		P_e
		\le
		\frac1{|\mathcal U_{N_n}|}
		\sum_{u\in\mathcal U_{N_n}}
		\sum_{(i,j)\in\mathcal I_d}
		\bbP_u\left[
		\bigl\langle \widehat u_n-u,e_i-e_j\bigr\rangle
		\ge \frac{1}{N_n}
		\right].
		\]
		For a message \(u\in\mathcal U_{N_n}\), under \(\bbP_u\), the output symbols
		\(Y_1,\ldots,Y_n\) are i.i.d. according to \(q_u=T(u)\).  Hence the empirical
		coordinate is
		\[
		\widehat u_n
		=
		H(\widehat P_{Y^n})
		=
		\frac1n\sum_{t=1}^n H(b_{Y_t}).
		\]
		For each \((i,j)\in\mathcal I_d\),
		\[
		\bigl\langle \widehat u_n-u,e_i-e_j\bigr\rangle
		=
		\frac1n\sum_{t=1}^n Z_{ij,u}(Y_t),
		\]
		where
		\(
		\mathbb E_{q_u}[Z_{ij,u}(Y)]=0\) and \(
		\operatorname{Var}_{q_u}[Z_{ij,u}(Y)]=V_{ij}(u)
		\).
		
		Lemma~\ref{lem:uniform-moments} gives a uniform lower bound on the variances and
		a uniform upper bound on the third absolute moments.  Hence the
		Berry-Esseen theorem \cite[Ch.~XVI.5 Theorem 2]{Feller_book}  applies uniformly over \(u\in K_W\) and
		\((i,j)\in\mathcal I_d\), giving
		\begin{align}
			&\bbP_u \left[
			\bigl\langle \widehat u_n-u,e_i-e_j\bigr\rangle
			\ge \frac{1}{N_n}
			\right]  \notag\\
			&\qquad\le
			\Phi\left(
			-\frac{\sqrt n}{{N_n}\sqrt{V_{ij}(u)}}
			\right)
			+
			O\left(\frac1{\sqrt n}\right),
			\label{eq:BE-transfer-refined}
		\end{align}
		where the \(O(n^{-1/2})\) term is uniform in \(u\) and \((i,j)\).
		
		Substituting \eqref{eq:BE-transfer-refined} into the finite-blocklength
		bound from Theorem~\ref{thm:finite-achievability} gives
		\begin{align}
			P_e
			&\le
			\frac1{|\mathcal U_{N_n}|}
			\sum_{u\in\mathcal U_{N_n}}
			\sum_{(i,j)\in\mathcal I_d}
			\Phi\left(
			-\frac{\sqrt n}{{N_n}\sqrt{V_{ij}(u)}}
			\right)  + O\left(\frac1{\sqrt n}\right).
			\label{eq:avg-error-rsum-before}
		\end{align}
		Because \({N_n}=\lfloor c\sqrt n\rfloor\),
		\[
		\frac{\sqrt n}{N_n}\to \frac1c .
		\]
		The summand in \eqref{eq:avg-error-rsum-before} is continuous on \(K_W\),
		uniformly bounded, and converges uniformly to
		\[
		\sum_{(i,j)\in\mathcal I_d}
		\Phi\left(
		-\frac{1}{c\sqrt{V_{ij}(u)}}
		\right).
		\]
		Since \(K_W\) is a fixed polytope with boundary of zero
		\(d\)-dimensional volume, the standard Riemann-sum convergence for the uniform grid \(\cL_{N_n}\) on the affine hyperplane \(\sum_i u_i=1\) gives
		\[
		\frac1{|\mathcal U_{N_n}|}
		\sum_{u\in\mathcal U_{N_n}}
		\sum_{(i,j)\in\mathcal I_d}
		\Phi\left(
		-\frac{\sqrt n}{{N_n}\sqrt{V_{ij}(u)}}
		\right)
		=
		A_W(c)+o(1).
		\]
		Thus
		\[
		P_e
		\le
		A_W(c)+o(1).
		\]
		By the definition of \(c_\epsilon\) and the monotonicity of \(A_W\),
		we have \(A_W(c)<\epsilon\). Hence
		\(P_e\le\epsilon\) for all sufficiently large \(n\). Therefore,
		\(M^\star(n,\epsilon)\ge |\mathcal U_{N_n}|\).
		It remains to count the messages.  By Lemma~\ref{lem:lattice-counting},
		\[
		|\mathcal U_{N_n}|
		=
		\lambda_W^\star\frac{{N_n}^d}{d!}
		+
		O_W({N_n}^{d-1}).
		\]
		Since \({N_n}=c\sqrt n(1+o(1))\),
		\[
		\log M^\star(n,\epsilon) 
		\ge
		\frac d2\log n+d\log c+\log\lambda_W^\star-\log d!+o(1).
		\]
		This completes the proof.
	\end{proof}
	
	\begin{remark}[Finite-lattice Gaussian approximation]
		For finite-blocklength numerical evaluation, it is useful to keep the
		lattice average before passing to the continuous limit.  For
		\(n,N\ge1\) with \(\mathcal U_N\ne\varnothing\), define
		\begin{equation}
			A_{n,N}^{\rm lat}
			:=
			\frac1{|\mathcal U_N|}
			\sum_{u\in\mathcal U_N}
			\sum_{(i,j)\in\mathcal I_d}
			\Phi\left(
			-\frac{\sqrt n}{N\sqrt{V_{ij}(u)}}
			\right).
			\label{eq:finite-lattice-gaussian-functional}
		\end{equation}
		By the same argument as in Proposition~\ref{prop:refined-average-achievability},
		any sequence \(N_n\to\infty\) satisfying
		\( 
		\limsup_{n\to\infty} A_{n,N_n}^{\rm lat}<\epsilon
		\)
		yields
		\[
		\log M^\star(n,\epsilon)
		\ge
		d\log N_n+\log\lambda_W^\star-\log d!+o(1).
		\]
		If \(N_n/\sqrt n\to c\in(0,\infty)\), then
		\(A_{n,N_n}^{\rm lat}\to A_W(c)\).
	\end{remark}
	\begin{remark} \label{remark:worst-case}
		The conservative inverse-normal Gaussian approximation is recovered from Proposition \ref{prop:refined-average-achievability} by replacing the local variances by their
		worst-case value.  Indeed, since \(V_{ij}(u)\le V_{\max}\),
		\[
		A_W(c)
		\le
		R_d\Phi\left(-\frac{1}{c\sqrt{V_{\max}}}\right).
		\]
		Thus any \(c\) strictly smaller than
		\( 
		1/{\sqrt{V_{\max}}(-\Phi^{-1}(\epsilon/R_d))}
		\)
		is admissible in Proposition~\ref{prop:refined-average-achievability}.
		Letting \(c\) approach this value gives the worst-case Gaussian
		achievability approximation. Up to bounded constants, this recovers the
		inverse-normal Gaussian approximation obtained in the full-dimensional analysis of
		\cite{feng_lower_isit_2025}. 
	\end{remark}
	
	\section{A Jeffreys-Mixture Refined Converse}
	\label{sec:jeffreys-mixture-converse}
	
	The converse in Section~\ref{sec:converse-covering} uses a finite KL covering
	of \(\cP_W\).  This is sufficient to identify the
	\(d\log\sqrt n\) growth term with a bounded remainder,
	but it does not identify the constant-order geometric structure of the
	auxiliary output distribution.  In this section we refine the auxiliary
	distribution in the meta-converse by replacing the finite covering mixture with a continuous mixture over the reachable output polytope \(\mathcal P_W\), equipped with its Fisher volume element on the affine hull.  This is
	analogous to Bayes-mixture asymptotics and minimax redundancy/regret arguments
	in smooth parametric families, where Jeffreys-type priors identify the
	Fisher-volume term \cite{clarke_barron_bayes_1990,takeuchi_barron_minimax_2024}.
	This continuous-mixture viewpoint complements the divergence-covering approach
	for noisy permutation channels in \cite{tang_capacity_2023}.

	\subsection{Jeffreys mixture and Fisher volume}
	Let
	\[
	\mathcal H:=\left\{z\in\bbR^m:\sum_{y\in\cY}z_y=0\right\}.
	\]
	Let
	\[
	\mathsf L_W:=\operatorname{span}(\mathcal P_W-\mathcal P_W)
	\subseteq \mathcal H,
	\qquad d=\dim\mathcal P_W .
	\]
	Choose an \(m\times d\) matrix \(A_W\) whose columns form an orthonormal basis
	of \(\mathsf L_W\).  For \(q\in\mathcal P_W\), define the Fisher information
	matrix on the tangent space of \(\mathcal P_W\) by
	\begin{equation}
		I_W(q)
		:=
		A_W^\top
		\diag\!\left(\frac1{q(y)}\right)
		A_W .
		\label{eq:jeffreys-fisher}
	\end{equation}
	The determinant is independent of the choice of orthonormal basis \(A_W\).
	Equivalently, \(\sqrt{\det I_W(q)}\,d\operatorname{vol}_d(q)\) is the
	Jeffreys volume element induced on the affine hull of \(\mathcal P_W\).
	
	By Assumption~\ref{ass:positive-support}, every \(q\in\mathcal P_W\) satisfies
	\(q(y)\ge p_{\min}\) for all \(y\).  Hence \(I_W(q)\) is continuous and
	uniformly nonsingular on \(\mathcal P_W\).  Define the Fisher volume
	\begin{equation}
		\mathcal J_W
		:=
		\int_{\mathcal P_W}
		\sqrt{\det I_W(q)}\,d\operatorname{vol}_d(q)
		<\infty .
		\label{eq:fisher-volume}
	\end{equation}
	We call \(\mathcal J_W\) the Fisher volume of the reachable output polytope.
	It is the Fisher-metric analogue of the Euclidean volume term
	\(\lambda_W^\star\) appearing in the simplex-lattice achievability and in the
	covering converse.
	
	Define the Jeffreys-mixture auxiliary output distribution over the
	\(d\)-dimensional reachable polytope by
	\begin{equation}
		Q_{\mathcal P_W,n}^{J}(y^n)
		:=
		\frac{1}{\mathcal J_W}
		\int_{\mathcal P_W}
		q^{\otimes n}(y^n)
		\sqrt{\det I_W(q)}\,d\operatorname{vol}_d(q).
		\label{eq:jeffreys-mixture-output}
	\end{equation}
	We also write \(Q_{Y^n}^{J}:=Q_{\mathcal P_W,n}^{J}\) if there is no confusion.
	
	\subsection{Local asymptotics}
	We use the notation \(p_n=\pi_{x^n}\), \(q_n=p_nW\), and
	\(P_{x^n}=P_{Y^n|x^n}^{\rm perm}\) from
	Section~\ref{sec:constant-remainder-converse}.  We first collect the definition used throughout this section.
	
	For \(p\in\Delta_{|\cX|-1}\), write \(q=pW\), and define
	\begin{align}
		\Sigma(p)
		&:=
		\sum_{x\in\cX}p(x)
		\left[
		\diag(W_x)-W_xW_x^\top
		\right],
		\label{eq:perm-count-covariance}
		\\
		\Sigma_{\rm iid}(q)
		&:=
		\diag(q)-qq^\top .
		\label{eq:iid-count-covariance}
	\end{align}
	Here \(W_x=W(\cdot|x)\).  The inverses and determinants of
	\(\Sigma(p)\) and \(\Sigma_{\rm iid}(q)\) are understood on \(\mathcal H\).
	
	For \(q\in\mathcal P_W\) and \(z\in\mathcal H\), define
	\begin{align}
		s_q(z)
		&:=
		A_W^\top
		\diag\!\left(\frac1{q(y)}\right)z,
		\label{eq:score-linear-term}
		\\
		\xi_q(z)
		&:=
		s_q(z)^\top I_W(q)^{-1}s_q(z).
		\label{eq:xi-definition}
	\end{align}
	The matrices \(I_W(q)\) are ordinary \(d\times d\) matrices in the
	orthonormal coordinates of \(\mathsf L_W\).
	
	For \(p\in\Delta_{|\cX|-1}\), with \(q=pW\), define
	\begin{align}
		R_p(z)
		&:=
		\frac12
		\log
		\frac{\det_{\mathcal H}\Sigma_{\rm iid}(q)}
		{\det_{\mathcal H}\Sigma(p)} 
		-
		\frac{\log e}{2}
		z^\top
		\left[
		\Sigma(p)^{-1}
		-
		\Sigma_{\rm iid}(q)^{-1}
		\right]z ,
		\label{eq:Rp-definition}
	\end{align}
	and
	\begin{equation}
		H_p(z)
		:=
		R_p(z)-\frac{\log e}{2}\xi_{pW}(z).
		\label{eq:Hp-definition}
	\end{equation}
	When \(Z_p\) appears below, it denotes a Gaussian random vector on
	\(\mathcal H\) with law
	\[
	Z_p\sim\mathcal N(0,\Sigma(p)).
	\]
	
	We next sharpen Lemma~\ref{lem:uniform-count-local} by applying a local central limit (CLT) theorem.
	
	\begin{lemma} 
		\label{lem:jeffreys-local-count-lr}
		With \(R_p\) defined in \eqref{eq:Rp-definition}, the following expansion
		holds uniformly over deterministic input sequences \(x^n\), over \(z\) in
		compact subsets of \(\mathcal H\), and over count vectors
		\(\nu=nq_n+\sqrt n z+O(1)\), where \(p_n=\pi_{x^n}\) and \(q_n=p_nW\):
		\begin{equation}
			\log
			\frac{dP_{x^n}}{dq_n^{\otimes n}}(y^n)
			=
			R_{p_n}(z)+o(1),
			\qquad N(y^n)=\nu.
			\label{eq:local-count-lr-expansion-uniform}
		\end{equation}
		Consequently, if \(p_n\to p\) and \(q=pW\), then under \(P_{x^n}\),
		\[
		Z_n:=\frac{N(Y^n)-nq_n}{\sqrt n}
		\Rightarrow
		Z_p\sim\mathcal N(0,\Sigma(p)),
		\]
		and
		\begin{equation}
			\log
			\frac{dP_{x^n}}{dq_n^{\otimes n}}(y^n)
			=
			R_p(z)+o(1),
			\qquad N(y^n)=nq_n+\sqrt n z+O(1),
			\label{eq:local-count-lr-expansion}
		\end{equation}
		uniformly for \(z\) in compact subsets of \(\mathcal H\).
	\end{lemma}
	\begin{proof}
		Under \(P_{x^n}=P_{Y^n|x^n}^{\rm perm}\), the output count vector is not
		affected by the final permutation.  Hence its law is the same as the count
		vector of independent outputs with single-letter laws
		\(W(\cdot|x_1),\ldots,W(\cdot|x_n)\).  Since \(W(y|x)\ge p_{\min}\), the
		hypotheses of Lemma~\ref{lem:appendix-uniform-count-local-clt} in Appendix~\ref{app:uniform-lattice-local-clt} hold uniformly
		over all deterministic input sequences.
		
		The average output law is \(q_n=p_nW\), and the average count covariance is
		\[
		\Sigma(p_n)
		=
		\sum_x p_n(x)
		\left[
		\diag(W_x)-W_xW_x^\top
		\right].
		\]
		Therefore, uniformly for \(z\) in compact subsets of \(\mathcal H\) and
		\(\nu=nq_n+\sqrt n z+O(1)\),
		\[
		P_{x^n}^{\rm cnt}(\nu)
		=
		\frac{\Delta_{\mathcal H}(1+o(1))}
		{(2\pi n)^{(m-1)/2}\sqrt{\det_{\mathcal H}\Sigma(p_n)}}
		\exp\left\{
		-\frac12 z^\top\Sigma(p_n)^{-1}z
		\right\}.
		\]
		
		The same lemma applies to the multinomial law
		\(\operatorname{Mult}(n,q_n)\), since \(q_n(y)\ge p_{\min}\).  Its per-symbol
		covariance is \(\Sigma_{\rm iid}(q_n)=\diag(q_n)-q_nq_n^\top\), so
		\[
		\begin{aligned}
			\operatorname{Mult}(n,q_n)(\nu)
		=
		\frac{\Delta_{\mathcal H}(1+o(1))}
		{(2\pi n)^{(m-1)/2}
			\sqrt{\det_{\mathcal H}\Sigma_{\rm iid}(q_n)}}\exp\left\{
		-\frac12 z^\top\Sigma_{\rm iid}(q_n)^{-1}z
		\right\}.
		\end{aligned}
		\]
		
		By Lemma~\ref{lem:count-space-likelihood-ratio}, applied with \(r=q_n\),
		for every \(y^n\) with \(N(y^n)=\nu\),
		\[
		\log\frac{dP_{x^n}}{dq_n^{\otimes n}}(y^n)
		=
		\log
		\frac{P_{x^n}^{\rm cnt}(\nu)}
		{\operatorname{Mult}(n,q_n)(\nu)} .
		\]
		The common lattice factor \(\Delta_{\mathcal H}\) and the common
		\((2\pi n)^{-(m-1)/2}\) factor cancel.  Taking the ratio and converting the
		exponential term to base-two logarithms gives
		\[
		\log
		\frac{dP_{x^n}}{dq_n^{\otimes n}}(y^n)
		=
		R_{p_n}(z)+o(1),
		\]
		uniformly over deterministic \(x^n\), compact \(z\)-sets, and count vectors
		\(\nu=nq_n+\sqrt n z+O(1)\).  This proves
		\eqref{eq:local-count-lr-expansion-uniform}.
		
		If \(p_n\to p\), then \(q_n\to q=pW\).  The covariance matrices are uniformly
		nonsingular on \(\mathcal H\); hence continuity of determinant and inverse
		gives
		\[
		R_{p_n}(z)\to R_p(z)
		\]
		uniformly for \(z\) in compact subsets of \(\mathcal H\).  Finally, the centered summands form a bounded triangular array, and their
		average covariance \(\Sigma(p_n)\) converges to \(\Sigma(p)\).  Hence the
		Lindeberg--Feller CLT gives
		\[
		Z_n=\frac{N(Y^n)-nq_n}{\sqrt n}
		\Rightarrow
		Z_p\sim\mathcal N(0,\Sigma(p)).
		\]
		This proves the consequent statement and completes the proof.
	\end{proof}
	
	We first prove a version that remains uniform when the center approaches the
	boundary, provided that the boundary is still far away on the \(1/\sqrt n\)
	Laplace scale.
	
	\begin{lemma} 
		\label{lem:jeffreys-laplace-interior-zone}
		With \(Q_{\mathcal P_W,n}^{J}\) defined in \eqref{eq:jeffreys-mixture-output}.  Fix \(L<\infty\), and let \(R_n\to\infty\).  Uniformly over
		all \(q\in\operatorname{relint}(\mathcal P_W)\) satisfying
		\[
		\sqrt n\,\operatorname{dist}(q,\relbd\mathcal P_W)\ge R_n ,
		\]
		and all count vectors satisfying
		\[
		N(y^n)=nq+\sqrt n z+O(1),
		\qquad \|z\|\le L ,
		\]
		we have
		\begin{align}
			\log\frac{q^{\otimes n}(y^n)}
			{Q_{\mathcal P_W,n}^{J}(y^n)}
			&=
			\frac d2\log n
			+
			\log\mathcal J_W
			-
			\frac d2\log(2\pi) 
			-
			\frac{\log e}{2}\xi_q(z)
			+
			o(1).
			\label{eq:jeffreys-laplace-interior-zone}
		\end{align}
	\end{lemma}
	
	\begin{proof}
		Let \(\mathsf L_W=\operatorname{span}(\mathcal P_W-\mathcal P_W)\), and let
		\(A_W\) be an orthonormal \(m\times d\) basis matrix for \(\mathsf L_W\).  We use the previously defined \(I_W(q)\), \(s_q(z)\), and \(\xi_q(z)\).  Since \(W\) is strictly
		positive, all \(q\in\mathcal P_W\) satisfy \(q(y)\ge p_{\min}\).  Hence
		\(I_W(q)\) is uniformly nonsingular and the Jeffreys density
		\(\sqrt{\det I_W(q)}\) is bounded above and below uniformly on \(\mathcal P_W\).
		
		Set
		\[
		v=q+\frac{A_Wh}{\sqrt n},
		\]
		and define the scaled feasible region
		\[
		D_n(q):=
		\left\{
		h\in\bbR^d:
		q+\frac{A_Wh}{\sqrt n}\in\mathcal P_W
		\right\}.
		\]
		If \(\delta(q):=\operatorname{dist}(q,\relbd\mathcal P_W)\), then
		\[
		B_d(0,\sqrt n\,\delta(q))\subseteq D_n(q).
		\]
		Therefore, under the assumption
		\(\sqrt n\,\delta(q)\ge R_n\to\infty\), every fixed ball
		\(B_d(0,M)\) is contained in \(D_n(q)\) for all sufficiently large \(n\),
		uniformly over the allowed \(q\)'s.
		
		Write \(N_y=nq(y)+\sqrt n z_y+b_y\), where \(\|z\|\le L\) and
		\(\|b\|=O(1)\).  For bounded \(h\), Taylor expansion of the natural logarithm
		gives, uniformly in \(q,z,h\),
		\[
		\sum_y N_y
		\ln\frac{q(y)+(A_Wh)_y/\sqrt n}{q(y)}
		=
		s_q(z)^\top h
		-
		\frac12 h^\top I_W(q)h
		+
		o(1).
		\]
		The term coming from \(nq\) has no first-order contribution because
		\(\sum_y(A_Wh)_y=0\).  The fluctuation term \(\sqrt n z\) gives
		\(s_q(z)^\top h\), and the second-order term from \(nq\) gives
		\(-\frac12 h^\top I_W(q)h\).  The \(O(1)\) rounding term contributes \(o(1)\)
		for bounded \(h\).
		
		Thus, on every fixed ball \(B_d(0,M)\),
		\[
		\frac{v^{\otimes n}(y^n)}{q^{\otimes n}(y^n)}
		=
		\exp\left\{
		s_q(z)^\top h-\frac12 h^\top I_W(q)h
		\right\}(1+o(1)),
		\]
		uniformly.  Also
		\[
		\sqrt{\det I_W(v)}
		=
		\sqrt{\det I_W(q)}+o(1)
		\]
		uniformly on bounded \(h\)-sets.
		
		It remains to justify that the contribution of \(\|h\|>M\) is negligible
		uniformly as \(M\to\infty\).  For \(v\in\mathcal P_W\),
		\[
		\sum_y nq(y)\ln\frac{v(y)}{q(y)}
		=
		-nD_{\rm e}(q\|v),
		\]
		where \(D_{\rm e}\) denotes natural-log KL divergence.  By Pinsker's inequality
		and strict positivity, there is \(c>0\) such that
		\[
		D_{\rm e}(q\|v)\ge c\|v-q\|^2
		\]
		uniformly for \(q,v\in\mathcal P_W\).  In the local region
		\(\|v-q\|\le\rho\), the fluctuation term is bounded by \(C_L\|h\|\), and hence
		the integrand is bounded by
		\[
		\exp\{-c\|h\|^2+C_L\|h\|+C\}.
		\]
		This gives a Gaussian tail bound, uniform in \(q\) and \(z\).  In the far region
		\(\|v-q\|\ge\rho\), compactness of \(\mathcal P_W\) and positivity of \(D_{\rm e}(q\|v)\) away from the diagonal give
		\(D_{\rm e}(q\|v)\ge c_\rho>0\), while the fluctuation term is only
		\(O(\sqrt n)\).  Hence that contribution is exponentially small in \(n\).
		Therefore the Laplace integral may be evaluated over all of \(\bbR^d\), with
		an \(o(1)\) relative error uniformly over the interior zone.
		
		Consequently,
		\[
		\begin{aligned}
			Q_{\mathcal P_W,n}^{J}(y^n)=
			\frac{q^{\otimes n}(y^n)}{\mathcal J_W}
			n^{-d/2}\sqrt{\det I_W(q)}
			\int_{\bbR^d}
			\exp\left\{
			s_q(z)^\top h-\frac12h^\top I_W(q)h
			\right\}dh
			(1+o(1)).
		\end{aligned}
		\]
		The Gaussian integral equals
		\[
		(2\pi)^{d/2}(\det I_W(q))^{-1/2}
		\exp\left\{\frac12\xi_q(z)\right\}.
		\]
		The Jeffreys factor \(\sqrt{\det I_W(q)}\) cancels the determinant factor, so
		\[
		Q_{\mathcal P_W,n}^{J}(y^n)
		=
		q^{\otimes n}(y^n)
		\frac{(2\pi)^{d/2}}{\mathcal J_W}
		n^{-d/2}
		\exp\left\{\frac12\xi_q(z)\right\}
		(1+o(1)).
		\]
		Taking base-two logarithms gives
		\eqref{eq:jeffreys-laplace-interior-zone}.
	\end{proof}
	
	The usual compact-interior Laplace approximation follows immediately.
	
	\begin{lemma}
		\label{lem:jeffreys-laplace}
		Let \(q\in\operatorname{relint}(\mathcal P_W)\).  For empirical fluctuations
		satisfying
		\[
		N(Y^n)=nq+\sqrt n z+O(1),
		\]
		with \(z\) in a compact subset of \(\mathcal H\),
		\begin{align}
			\log \frac{q^{\otimes n}(Y^n)}
			{Q_{\mathcal P_W,n}^{J}(Y^n)}
			=
			\frac d2\log n
			+
			\log\mathcal J_W
			-
			\frac d2\log(2\pi)
			-
			\frac{\log e}{2}\xi_q(z)
			+
			o(1).
			\label{eq:jeffreys-laplace-expansion}
		\end{align}
		The \(o(1)\) term is uniform for \(q\) in compact subsets of
		\(\operatorname{relint}(\mathcal P_W)\) and \(z\) in compact subsets of
		\(\mathcal H\).
	\end{lemma}
	
	\begin{proof}
		It suffices to prove the asserted uniformity with \(q\) restricted to an
		arbitrary compact set \(K\subset\operatorname{relint}(\mathcal P_W)\) and
		\(z\) restricted to an arbitrary compact set \(B\subset\mathcal H\).  Since
		\(K\) is compact and lies in the relative interior of \(\mathcal P_W\),
		\[
		\delta_K
		:=
		\inf_{q\in K}\operatorname{dist}(q,\relbd\mathcal P_W)
		>0 .
		\]
		Let \(L:=\sup_{z\in B}\|z\|<\infty\), and choose any sequence
		\(R_n\to\infty\) satisfying \(R_n\le \sqrt n\,\delta_K\) for all sufficiently
		large \(n\), for instance \(R_n=\log n\).  Then, uniformly over \(q\in K\),
		\[
		\sqrt n\,\operatorname{dist}(q,\relbd\mathcal P_W)
		\ge R_n
		\]
		for all sufficiently large \(n\).  Applying
		Lemma~\ref{lem:jeffreys-laplace-interior-zone} with this \(L\) and \(R_n\)
		gives
		\[
		\log \frac{q^{\otimes n}(Y^n)}
		{Q_{\mathcal P_W,n}^{J}(Y^n)}
		=
		\frac d2\log n
		+
		\log\mathcal J_W
		-
		\frac d2\log(2\pi)
		-
		\frac{\log e}{2}\xi_q(z)
		+
		o(1),
		\]
		uniformly for \(q\in K\), \(z\in B\), and count vectors satisfying
		\(N(Y^n)=nq+\sqrt n z+O(1)\).  This is exactly the claimed compact-interior
		statement.
	\end{proof}
	
	Combining the likelihood expansion with the Jeffreys-mixture
	Laplace expansion gives the following information-density limit.
	
	\begin{theorem}
		\label{thm:jeffreys-information-density}
		Let \(x^n\in\cX^n\) be deterministic input sequences with
		\(p_n:=\pi_{x^n}\to p\).  Set \(q=pW\), and suppose
		\(q\in\operatorname{relint}(\mathcal P_W)\).  Set \(q_n:=p_nW\).
		Then
		\begin{equation}
			\log
			\frac{dP_{x^n}}{dQ_{\mathcal P_W,n}^{J}}(Y^n)
			-
			\frac d2\log n
			-
			\log\mathcal J_W
			+
			\frac d2\log(2\pi)
			\Rightarrow
			H_p(Z_p),
			\label{eq:jeffreys-info-density-weak-limit}
		\end{equation}
		where \(Z_p\sim\mathcal N(0,\Sigma(p))\) on \(\mathcal H\), and \(H_p\) is
		defined in \eqref{eq:Hp-definition}.
	\end{theorem}
	\begin{proof}
		Since \(q_n\to q\) and \(q\in\operatorname{relint}(\mathcal P_W)\), we have
		\(q_n\in\operatorname{relint}(\mathcal P_W)\) for all sufficiently large \(n\).
		Under \(P_{x^n}\), define
		\[
		Z_n:=\frac{N(Y^n)-nq_n}{\sqrt n}.
		\]
		Decompose
		\[
		\log\frac{dP_{x^n}}{dQ_{\mathcal P_W,n}^{J}}(Y^n)
		=
		\log\frac{dP_{x^n}}{dq_n^{\otimes n}}(Y^n)
		+
		\log\frac{q_n^{\otimes n}(Y^n)}
		{Q_{\mathcal P_W,n}^{J}(Y^n)} .
		\]
		
		Fix \(L<\infty\), and define
		\[
		E_{n,L}:=\{\|Z_n\|\le L\}.
		\]
		Since \(Z_n\Rightarrow Z_p\), the sequence \(\{Z_n\}\) is tight.  Hence, for
		every \(\delta>0\), \(L\) can be chosen large enough so that
		\[
		\limsup_{n\to\infty} P_{x^n}(E_{n,L}^c)\le \delta .
		\]
		Moreover, since \(q_n\to q\in\operatorname{relint}(\mathcal P_W)\), the points
		\(q_n\) eventually belong to a compact subset of
		\(\operatorname{relint}(\mathcal P_W)\).
		
		On the event \(E_{n,L}\), the variable \(Z_n\) lies in the compact set
		\(\{z\in\mathcal H:\|z\|\le L\}\).  Therefore the compact-uniform expansions in
		Lemma~\ref{lem:jeffreys-local-count-lr} and
		Lemma~\ref{lem:jeffreys-laplace} apply uniformly with \(z=Z_n\) and \(q=q_n\).
		Thus, on \(E_{n,L}\),
		\[
		\log\frac{dP_{x^n}}{dq_n^{\otimes n}}(Y^n)
		=
		R_{p_n}(Z_n)+o(1),
		\]
		and
		\[
		\begin{aligned}
			\log\frac{q_n^{\otimes n}(Y^n)}
			{Q_{\mathcal P_W,n}^{J}(Y^n)}
			=
			\frac d2\log n
			+
			\log\mathcal J_W
			-
			\frac d2\log(2\pi)
			-
			\frac{\log e}{2}\xi_{q_n}(Z_n)
			+
			o(1),
		\end{aligned}
		\]
		where the two \(o(1)\) terms are uniform on \(E_{n,L}\).
		
		Since \(L\) can be chosen so that \(P_{x^n}(E_{n,L}^c)\) is arbitrarily small,
		the preceding two displays imply the corresponding \(o_{P_{x^n}}(1)\)
		expansions:
		\[
		\log\frac{dP_{x^n}}{dq_n^{\otimes n}}(Y^n)
		=
		R_{p_n}(Z_n)+o_{P_{x^n}}(1),
		\]
		and
		\[
		\begin{aligned}
			\log\frac{q_n^{\otimes n}(Y^n)}
			{Q_{\mathcal P_W,n}^{J}(Y^n)}
			=
			\frac d2\log n
			+
			\log\mathcal J_W
			-
			\frac d2\log(2\pi)
			-
			\frac{\log e}{2}\xi_{q_n}(Z_n)
			+
			o_{P_{x^n}}(1).
		\end{aligned}
		\]
		Finally, \(p_n\to p\), \(q_n\to q=pW\), and the tightness of \(Z_n\) give
		\[
		R_{p_n}(Z_n)=R_p(Z_n)+o_{P_{x^n}}(1),
		\qquad
		\xi_{q_n}(Z_n)=\xi_q(Z_n)+o_{P_{x^n}}(1).
		\]
		
		Adding the two expansions yields
		\[
		\begin{aligned}
			\log
			\frac{dP_{x^n}}{dQ_{\mathcal P_W,n}^{J}}(Y^n)
			&=
			\frac d2\log n
			+
			\log\mathcal J_W
			-
			\frac d2\log(2\pi) \\
			&\quad
			+
			H_p(Z_n)
			+
			o_{P_{x^n}}(1),
		\end{aligned}
		\]
		where
		\[
		H_p(z)=R_p(z)-\frac{\log e}{2}\xi_{pW}(z).
		\]
		Since \(Z_n\Rightarrow Z_p\) and \(H_p\) is continuous, the continuous mapping
		theorem gives \eqref{eq:jeffreys-info-density-weak-limit}.
	\end{proof}
	
	\subsection{Testing bound and converse}
	The fixed-dimensional testing term is defined through the limiting Gaussian
	experiment.  In this limiting problem, a randomized test is represented by a
	measurable function \(\varphi:\mathcal H\to[0,1]\).
	
	For \(\alpha\in(0,1)\) and \(p\in\Delta_{|\cX|-1}\), define
	\begin{equation}
		B_\alpha(p)
		:=
		\inf_{\substack{\varphi:\mathcal H\to[0,1]\\
				\mathbb E[\varphi(Z_p)]\ge\alpha}}
		\mathbb E\!\left[
		2^{-H_p(Z_p)}\varphi(Z_p)
		\right],
		\label{eq:Balpha-definition}
	\end{equation}
	where \(Z_p\sim\mathcal N(0,\Sigma(p))\) on \(\mathcal H\), and \(H_p\) is
	defined in \eqref{eq:Hp-definition}.
	
	The quantity \(B_\alpha(p)\) comes from the interior information-density limit
	in Theorem~\ref{thm:jeffreys-information-density}.  However, the meta-converse
	requires a uniform bound over all codewords, including those whose induced
	output types approach the relative boundary of \(\mathcal P_W\).  To prevent
	such boundary types from determining the worst-codeword supremum, we use a
	stratified auxiliary distribution that assigns a small amount of mass to
	Jeffreys mixtures on all proper faces of \(\mathcal P_W\).
	
	Let \(\mathfrak F_\partial\) be the finite collection of nonempty proper faces
	of \(\mathcal P_W\).  For each \(F\in\mathfrak F_\partial\), let
	\[
	\mathsf L_F:=\operatorname{span}(F-F),
	\qquad
	r_F:=\dim F .
	\]
	If \(r_F\ge1\), choose an orthonormal basis matrix \(A_F\) for
	\(\mathsf L_F\), define
	\[
	I_F(q):=
	A_F^\top \diag\!\left(\frac1{q(y)}\right) A_F,
	\qquad q\in F ,
	\]
	set
	\[
	\mathcal J_F:=
	\int_F \sqrt{\det I_F(q)}\,d\operatorname{vol}_{r_F}(q),
	\]
	where volume is taken in \(\aff(F)\), and define
	\[
	Q_{F,n}^{J}(y^n)
	:=
	\frac1{\mathcal J_F}
	\int_F q^{\otimes n}(y^n)
	\sqrt{\det I_F(q)}\,d\operatorname{vol}_{r_F}(q).
	\]
	If \(r_F=0\), say \(F=\{q_F\}\), set
	\[
	\mathcal J_F:=1,
	\qquad
	Q_{F,n}^{J}(y^n):=q_F^{\otimes n}(y^n).
	\]
	
	Fix weights \(\omega_F>0\), \(F\in\mathfrak F_\partial\), with
	\(\sum_{F\in\mathfrak F_\partial}\omega_F=1\).  Choose
	\(\gamma\in(0,1/2)\) and set \(\delta_n:=n^{-\gamma}\).  The stratified
	Jeffreys auxiliary output distribution is
	\begin{equation}
	Q_{Y^n}^{\rm str}
	:=
	(1-\delta_n)Q_{\mathcal P_W,n}^{J}
	+
	\delta_n
	\sum_{F\in\mathfrak F_\partial}\omega_F Q_{F,n}^{J}, \label{eq:jeffreys_aux}
	\end{equation}
	where \(Q_{\mathcal P_W,n}^{J}\) is the Jeffreys-mixture auxiliary output
	distribution
	defined in \eqref{eq:jeffreys-mixture-output}.
	
	As in auxiliary-output mixture constructions for finite-blocklength converses
	\cite{tomamichel_tight_2013}, the mixture is designed to control different
	classes of codewords at the correct polynomial scale.  The component
	\(Q_{\mathcal P_W,n}^{J}\) in the first term of \eqref{eq:jeffreys_aux}
	handles codewords whose induced output types remain in the interior of
	\(\mathcal P_W\), while the lower-dimensional face components reserve enough
	mass for boundary types.  Since \(\delta_n=n^{-\gamma}\) is vanishing but only
	polynomially small, it does not affect the interior constant-order term and
	still controls the boundary contribution.
	
	For \(\alpha\in(0,1)\), define
	\[
	C_{W,\alpha}^{J,{\rm int}}
	:=
	-\frac d2\log(2\pi)
	+
	\sup_{p\in\Delta_{|\cX|-1}}
	[-\log B_\alpha(p)] .
	\]
	This constant is finite.  Indeed, strict positivity of \(W\) gives uniform
	upper and lower spectral bounds for \(\Sigma(p)\), \(\Sigma_{\rm iid}(pW)\),
	and \(I_W(pW)\).  Hence \(H_p(z)\) has at most quadratic growth uniformly in
	\(p\), while the Gaussian laws \(Z_p\) are uniformly tight.  Therefore
	\(\inf_p B_\alpha(p)>0\).
	
	The next result formalizes this split between interior and boundary types.
	\begin{proposition}[Uniform testing bound for the stratified mixture]
		\label{prop:stratified-uniform-testing-bound}
		Under Assumption~\ref{ass:positive-support}, let
		\(d=\dim\mathcal P_W\ge1\). For every \(\alpha\in(0,1)\),
		\[
		\begin{aligned}
		&\limsup_{n\to\infty}
		\left[
		\sup_{x^n\in\cX^n}
		-\log\beta_\alpha
		\left(
		P_{Y^n|x^n}^{\rm perm},
		Q_{Y^n}^{\rm str}
		\right) -
		\frac d2\log n
		\right] \nonumber \\
		&\le
		C_{W,\alpha}^{J,{\rm int}} 
		+
		\log\mathcal J_W .
		\end{aligned}
		\]
	\end{proposition}
	\begin{proof}
		See Appendix \ref{apx:proof_stratified-uniform-testing-bound}
	\end{proof}
	
	Then, the main result in this section is the following
	
	\begin{theorem}
		\label{thm:stratified-jeffreys-refined-converse}
		Under Assumption~\ref{ass:positive-support}, let \(d=\dim\mathcal P_W\ge1\).
		For every \(\epsilon\in(0,1)\),
		\[
		\log M^\star(n,\epsilon)
		\le
		\frac d2\log n
		+
		\log\mathcal J_W
		+
		C_{W,\epsilon}^{J,{\rm str}}
		+
		o(1),
		\]
		where
		\[
		C_{W,\epsilon}^{J,{\rm str}}
		:=
		\inf_{\eta\in(\epsilon,1)}
		\left\{
		C_{W,1-\eta}^{J,{\rm int}}
		-
		\log\left(1-\frac{\epsilon}{\eta}\right)
		\right\}
		\]
		is finite.
	\end{theorem}
	
	\begin{proof}
		Apply Lemma~\ref{lem:symbol-relaxed-meta-converse-perm} with
		\(Q_{Y^n}=Q_{Y^n}^{\rm str}\).  For every \(\eta\in(\epsilon,1)\), set
		\(\alpha=1-\eta\).  Then
		\[
		\log M^\star(n,\epsilon)
		\le
		\sup_{x^n\in\cX^n}
		-\log\beta_\alpha
		\left(
		P_{Y^n|x^n}^{\rm perm},
		Q_{Y^n}^{\rm str}
		\right)
		-
		\log\left(1-\frac{\epsilon}{\eta}\right).
		\]
		Using Proposition~\ref{prop:stratified-uniform-testing-bound} gives
		\[
		\log M^\star(n,\epsilon)
		\le
		\frac d2\log n
		+
		\log\mathcal J_W
		+
		C_{W,1-\eta}^{J,{\rm int}}
		-
		\log\left(1-\frac{\epsilon}{\eta}\right)
		+
		o(1).
		\]
		Since this holds for every fixed \(\eta\in(\epsilon,1)\), optimizing over
		\(\eta\) gives the result.
	\end{proof}
	
	\begin{remark}
		\label{remark:gaussian-jeffreys-comparison}
			Proposition~\ref{prop:refined-average-achievability} gives, for every
		\(c<c_\epsilon\),
		\[
		\log M^\star(n,\epsilon)
		\ge
		d\log(c\sqrt n)+\log\lambda_W^\star-\log d!+o(1).
		\]
		Equivalently, by letting \(c\uparrow c_\epsilon\) after taking the
		liminf,
		\[
		\liminf_{n\to\infty}
		\left[
		\log M^\star(n,\epsilon)
		-
		d\log(c_\epsilon\sqrt n)
		-
		\log\lambda_W^\star
		\right]
		\ge
		-\log d! .
		\]
		On the converse side, Theorem~\ref{thm:stratified-jeffreys-refined-converse}
		gives
		\[
		\log M^\star(n,\epsilon)
		\le
		d\log\sqrt n
		+
		\log\mathcal J_W
		+
		C_{W,\epsilon}^{J,{\rm str}}
		+
		o(1).
		\]
		Thus the refined achievability and converse approximations share the same
		affine-dimensional blocklength term \(d\log\sqrt n\).  Their constant-order
		geometric terms arise from different natural constructions: the Euclidean
		reference-simplex volume ratio \(\lambda_W^\star\) on the achievability side,
		and the Fisher volume \(\mathcal J_W\) on the Jeffreys-mixture converse side.
	\end{remark}

	\section{Numerical Results} \label{sec:numerical_results}
	\begin{figure*}[t]
		\normalsize	
		\centering
		\begin{minipage}[t]{0.4 \linewidth}
			\centering
			\includegraphics[width = 1\textwidth]{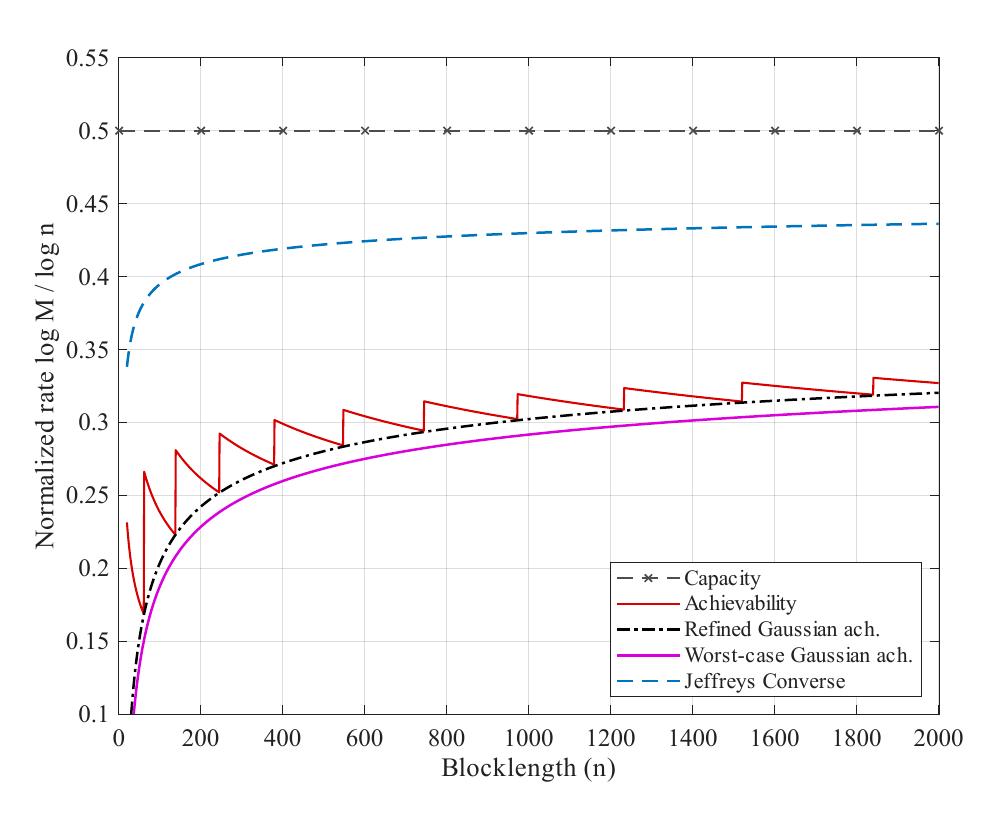}
			\captionsetup{font=footnotesize}
			\caption{\  Normalized rate \( \log M/ \log n\) versus blocklength for the BSC with crossover probability $\delta=0.11$ and average error probability $\epsilon = 10^{-3}$.}
			\label{fig:BSC_11_e-3_approximation}
		\end{minipage} 
		\hspace{40pt}
		\begin{minipage}[t]{0.4 \linewidth}
			\centering
			\includegraphics[width = 1\textwidth]{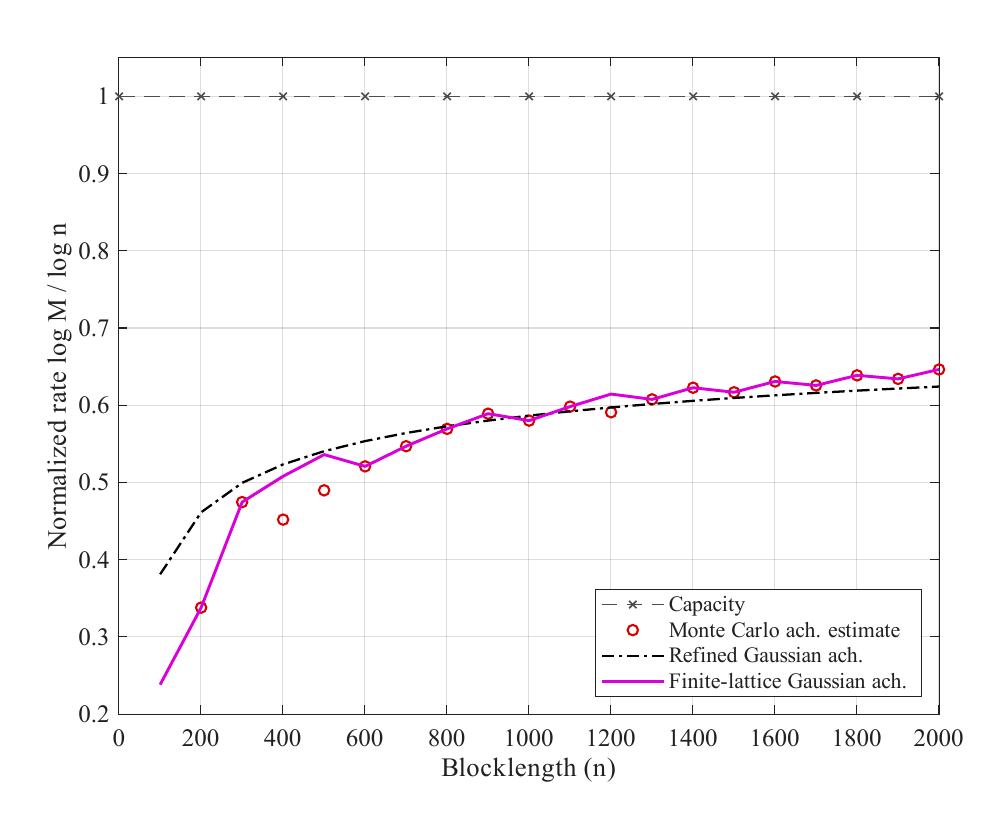}
			\captionsetup{font=footnotesize}
			\caption{\  Normalized rate \( \log M/ \log n\) versus blocklength for the given \(3 \times 4\) channel with average error probability $\epsilon = 10^{-2}$: comparison of the Gaussian approximations and Monte Carlo estimate.}
			\label{fig:34_DMC}
		\end{minipage}
	\end{figure*}

	\subsection{Binary Case}

	We first revisit the binary symmetric channel (BSC) with crossover probability
	\(\delta\). This one-dimensional example provides a simple reference case for
	comparing the Gaussian approximations. Since the BSC is one-dimensional and admits
	explicit maximum-likelihood comparisons between adjacent grid points, we
	use the non-asymptotic achievability bound from
	\cite{feng_lower_isit_2025} as the finite-blocklength bound.  Our focus in this example is therefore on the comparison among the
	finite-blocklength achievability bound, the two Gaussian achievability
	approximations, and the Jeffreys-mixture converse approximation.

	In this case \(d=1\), the two ordered transfer directions correspond to the two
	nearest-neighbor errors, and the reachable output set is
	\(
	[\delta,1-\delta]\subseteq\Delta_1 .
	\)
	Hence the relative length is \(\lambda_W^\star=1-2\delta\).
	
	Write \(c_{\delta,\epsilon}:=(1-2\delta)c_\epsilon\).  Specializing the
	average Gaussian union equation \(A_W(c_\epsilon)=\epsilon\) to the BSC gives
	\begin{equation}
		\frac{2}{1-2\delta}
		\int_{\delta}^{1-\delta}
		\Phi\left(
		-\frac{1-2\delta}
		{2c_{\delta,\epsilon}\sqrt{q(1-q)}}
		\right)dq
		=
		\epsilon .\nonumber
	\end{equation}
	Since \(d=1\), \(\log d!=0\), and \(\lambda_W^\star=1-2\delta\),
	Proposition~\ref{prop:refined-average-achievability} gives, after letting
	\(c\uparrow c_\epsilon\),
	\begin{equation}
		\log M^\star(n,\epsilon)
		\ge
		\frac12\log n+\log c_{\delta,\epsilon}+o(1).
		\label{eq:BSC_approx_refined}
	\end{equation}
	We plot the right-hand side without the \(o(1)\) term, normalized by
	\(\log n\),  as the refined Gaussian
	achievability approximation.
	
	For comparison, the worst-case Gaussian achievability approximation obtained
	from the argument in Remark~\ref{remark:worst-case} recovers the expression in
	\cite[Corollary~1]{feng_lower_isit_2025}:
	\begin{equation}
		\log M^\star(n,\epsilon)
		\gtrsim
		\log \frac{(1-2\delta)\sqrt n}{-\Phi^{-1}(\epsilon/2)} .
		\label{eq:BSC_approx_worst}
	\end{equation}
	In the figure, we plot the right-hand side of
	\eqref{eq:BSC_approx_worst} normalized by \(\log n\).
	
	For the Jeffreys-mixture converse approximation, the Fisher volume of the BSC
	reachable interval is
	\[
	\mathcal J_W
	=
	\int_\delta^{1-\delta}\frac{dq}{\sqrt{q(1-q)}}
	=
	\pi-4\arcsin\sqrt{\delta}.
	\]
	The finite-dimensional testing constant can also be evaluated explicitly.  Let
	\(v_\delta:=\delta(1-\delta)\).  A direct specialization of the Gaussian
	testing functional gives, for \(\alpha\in(0,1)\),
	\[
	C_{W,\alpha}^{J,{\rm int}}
	=
	-\log\!\left[
	4\sqrt{v_\delta}\,
	\Phi^{-1}\!\left(\frac{1+\alpha}{2}\right)
	\right].
	\]
	Therefore,
	\[
	C_{W,\epsilon}^{J,{\rm str}}
	=
	\inf_{\eta\in(\epsilon,1)}
	\left\{
	-\log\!\left[
	4\sqrt{v_\delta}\,
	\Phi^{-1}\!\left(1-\frac{\eta}{2}\right)
	\right]
	-
	\log\left(1-\frac{\epsilon}{\eta}\right)
	\right\}.
	\]
	Theorem~\ref{thm:stratified-jeffreys-refined-converse}
	then gives
	\begin{equation}
		\log M^\star(n,\epsilon)
		\le
		\frac12\log n+\log\mathcal J_W+C_{W,\epsilon}^{J,{\rm str}}+o(1).
		\label{eq:BSC_jef}
	\end{equation}
	We plot the right-hand side without the \(o(1)\) term, normalized by
	\(\log n\),  as the Jeffreys-mixture converse approximation.
	
	Fig.~\ref{fig:BSC_11_e-3_approximation} compares the non-asymptotic
	achievability bound, the two Gaussian achievability approximations, and the
	Jeffreys-mixture converse approximation.  The refined Gaussian achievability
	approximation in \eqref{eq:BSC_approx_refined} improves the worst-case
	Gaussian achievability approximation in \eqref{eq:BSC_approx_worst}: by
	averaging the local variance over the reachable output interval, it reduces
	the conservativeness of the worst-case variance bound and better follows the
	smooth trend of the non-asymptotic achievability curve.  The remaining
	oscillations in the latter curve are caused by finite lattice-resolution
	effects and are not modeled by the continuous Gaussian approximation.  The
	Jeffreys-mixture converse approximation lies above the achievability curves
	and approaches the same normalized logarithmic capacity \(d/2=1/2\).

	\subsection{Lower-Dimensional Case}
	We next consider a lower-dimensional example, where the reachable output
	polytope occupies a proper affine slice of the output simplex.  The purpose of
	this example is to illustrate how the simplex-lattice
	construction and the refined Gaussian achievability approximation behave when
	the reachable polytope is genuinely lower-dimensional.
	
	For the numerical evaluation, we consider the strictly positive \(3\times4\)
	channel
	\[
	W=
	\begin{pmatrix}
		0.80 & 0.05 & 0.05 & 0.10\\
		0.05 & 0.80 & 0.05 & 0.10\\
		0.05 & 0.05 & 0.80 & 0.10
	\end{pmatrix},
	\]
	whose reachable output polytope is two-dimensional.  For each lattice
	resolution \(N\), the message set is
	\[
	\cU_N=K_W\cap\cL_N .
	\]
	For this channel,
	\(
	\cR_W
	=
	\left\{
	q\in\Delta_3:\ q_4=0.1
	\right\},
	\)
	which is itself a \(2\)-simplex.  Hence \(S_W^\star=\cR_W\).  We use the
	affine parametrization
	\(
	T(u_1,u_2,u_3)
	=
	(0.9u_1,0.9u_2,0.9u_3,0.1),
	\qquad u\in\Delta_2 .
	\)
	Here $d=2$,
	\[
	K_W=
	\left\{
	u\in\Delta_2:\ u_i\ge \frac1{18},\ i=1,2,3
	\right\},
	\]
	and
	\[
	\lambda_W^\star
	=
	\frac{\operatorname{vol}_2(K_W)}
	{\operatorname{vol}_2(\Delta_2)}
	=
	\frac{25}{36}.
	\]
	For this example, we use the coordinate map
	\[
	\widetilde H(p)
	=
	\left(
	\frac{p_1}{0.9},
	\frac{p_2}{0.9},
	\frac{p_3}{0.9}
	\right).
	\]
	This map gives the same nearest-neighbor decisions and the same transfer
	fluctuations as the projection-based coordinate map in \eqref{eq:H-projection}.
	Therefore,
	\(
	V_{ij}(u)
	=
	\frac{u_i+u_j}{0.9}
	-
	(u_i-u_j)^2\) for every \( i\ne j\), where  \(u=(u_1,u_2,u_3)\in K_W\). The coefficient \(c_\epsilon\)
	is determined by the average Gaussian union function
	\begin{equation}
		\frac{1}{\operatorname{vol}_2(K_W)}
		\int_{K_W}
		\sum_{i\ne j}
		\Phi\left(
		-\frac{1}{c_\epsilon\sqrt{
				\frac{u_i+u_j}{0.9}-(u_i-u_j)^2}}
		\right)
		du
		=
		\epsilon .\nonumber
	\end{equation}
	Following Proposition~\ref{prop:refined-average-achievability}, the refined
	Gaussian achievability bound gives
	\begin{equation}
		\log M^\star(n,\epsilon)
		\ge
		2\log(c_\epsilon\sqrt n)
		+
		\log\frac{25}{36}
		-
		\log2
		+
		o(1).
		\label{eq:3x4-continuous-refined}
	\end{equation}
	We plot the right-hand side without the \(o(1)\) term, normalized by
	\(\log n\), as the refined Gaussian achievability approximation.
	
	For the finite-lattice Gaussian approximation, define
	\begin{equation}
		A_{n,N}^{\rm lat}
		:=
		\frac1{|\cU_N|}
		\sum_{u\in\cU_N}
		\sum_{i\ne j}
		\Phi\left(
		-\frac{\sqrt n}
		{N\sqrt{\frac{u_i+u_j}{0.9}-(u_i-u_j)^2}}
		\right).
		\nonumber
	\end{equation}
	For each blocklength \(n\), let
	\[
	N_G(n,\epsilon)
	:=
	\max\left\{
	N:\ A_{n,N}^{\rm lat}\le\epsilon
	\right\}.
	\]
	The Gaussian surrogate gives the finite-lattice achievability approximation
	\begin{equation}
		\log|\cU_{N_G(n,\epsilon)}|.
		\label{eq:finite-lat}
	\end{equation}
	This quantity is plotted as a Gaussian achievability approximation to
	\(\log M^\star(n,\epsilon)\), normalized by \(\log n\).

	For the Monte Carlo achievability estimate, we simulate the construction in
	Section~\ref{sec:achievability}. Each message
	\(u\in\cU_N\) induces
	\[
	q_u=T(u)=(0.9u_1,0.9u_2,0.9u_3,0.1).
	\]
	Given an output sequence \(Y^n\), we form the empirical output distribution
	\(\widehat q\) and compute
	\[
	\widehat u(Y^n)
	=
	\left(
	\frac{\widehat q_1}{0.9},
	\frac{\widehat q_2}{0.9},
	\frac{\widehat q_3}{0.9}
	\right).
	\]
	The decoder applies \(g_N\) to \(\widehat u(Y^n)\), with ties counted as
	errors.
	
	For fixed \(n\) and \(N\), the average error probability is
	\[
	P_{\rm e}(n,N)
	=
	\frac1{|\cU_N|}
	\sum_{u\in\cU_N}
	\mathbb P_{Y^n\sim q_u^{\otimes n}}
	\left[
	g_N(\widehat u(Y^n))\ne u
	\right].
	\]
	We estimate \(P_{\rm e}(n,N)\) by Monte Carlo simulation.  For each tested
	pair \((n,N)\) and each message \(u\in\cU_N\), we generate
	\(T_{\rm MC}=300\) independent count vectors
	\[
	C_{u,\ell}\sim \operatorname{Mult}(n,q_u),
	\qquad \ell=1,\ldots,T_{\rm MC}.
	\]
	This is equivalent to sampling \(Y^n\sim q_u^{\otimes n}\) and retaining only
	its empirical distribution.  Let
	\[
	\widehat q_{u,\ell}:=\frac{C_{u,\ell}}{n},
	\qquad
	\widehat u_{u,\ell}:=\widetilde H(\widehat q_{u,\ell}).
	\]
	The Monte Carlo estimate of the average error probability is
	\[
	\widehat P_{\rm e}(n,N)
	=
	\frac1{|\cU_N|}
	\sum_{u\in\cU_N}
	\frac1{T_{\rm MC}}
	\sum_{\ell=1}^{T_{\rm MC}}
	\mathbf 1\!\left\{
	g_N(\widehat u_{u,\ell})\ne u
	\right\}.
	\]
	For each \(n\), among the tested lattice resolutions, let
	\[
	N_{\rm MC}(n,\epsilon)
	:=
	\max\left\{
	N:\ \widehat P_{\rm e}(n,N)\le\epsilon
	\right\}.
	\]
	The Monte Carlo simulation then gives the empirical achievability estimate
	\begin{equation}
		\log|\cU_{N_{\rm MC}(n,\epsilon)}|.
		\label{eq:non_asyp_b}
	\end{equation}
	Equivalently, this is plotted as an achievability estimate for
	\(\log M^\star(n,\epsilon)\), normalized by \(\log n\).
	
	Fig.~\ref{fig:34_DMC} compares the Monte Carlo achievability estimate, the
	finite-lattice Gaussian approximation in \eqref{eq:finite-lat}, and the refined
	Gaussian achievability approximation in \eqref{eq:3x4-continuous-refined}.  The
	finite-lattice Gaussian approximation closely tracks the Monte Carlo estimate
	over most displayed blocklengths, suggesting that the local Gaussian surrogate
	captures the finite-resolution behavior of the simplex-lattice construction.
	The refined Gaussian achievability approximation provides a smoother asymptotic
	trend and does not model the integer effects of the lattice resolution.

	\section{Conclusion}
	
	We studied volume-refined fixed-error achievability and converse bounds for
	strictly positive noisy permutation channels through the affine geometry of the
	reachable output polytope.  By
	working on the affine hull of \(\cP_W\), the code construction, decoder, and
	error analysis depend on the intrinsic dimension \(d=\dim\cP_W\), rather than
	on the ambient output-simplex dimension \(|\cY|-1\).
	
	On the achievability side, the affine-coordinate simplex-lattice construction
	and the error-reduction lemma reduce nearest-neighbor decoding errors to
	\(d(d+1)\) one-dimensional transfer events.  This yields a Gaussian
	achievability approximation whose coefficient is determined by averaged local
	coordinate variances over the reachable output polytope.  On the converse side,
	the meta-converse combined with an affine-dimensional divergence covering and a
	local testing estimate gives a bounded-remainder fixed-error converse.  Together
	with the achievability bound, this establishes the fixed-error logarithmic
	capacity.
	
	We further refined the converse approximation by using a stratified
	Jeffreys-mixture auxiliary output distribution.  The component over
	\(\cP_W\) identifies the Fisher-volume term \(\log\mathcal J_W\) through a
	local Laplace approximation, while the lower-dimensional face components make
	the bound uniform over boundary output types.  This gives a constant-order
	converse approximation with an explicit Gaussian testing constant and an
	\(o(1)\) remainder, complementing the Gaussian achievability approximation.

	\appendices

	\section{Proof of Lemma \ref{lem:uniform-count-local}}
	\label{apx:proof_uniform-count-local}
	
	The concentration bound follows from Chebyshev's inequality.  Since \(S_n=\sum_{s=1}^n B_s\) and \(B_1,\ldots,B_n\) are independent,
	\[
	\operatorname{tr}\operatorname{Cov}(S_n)
	=
	\sum_{s=1}^n \operatorname{tr}\operatorname{Cov}(B_s).
	\]
	For each \(s\), \(B_s\) takes values in the standard basis vectors, and hence
	\[
	\operatorname{tr}\operatorname{Cov}(B_s)
	=
	\mathbb E\|B_s\|_2^2-\|\mathbb E B_s\|_2^2
	=
	1-\|P_s\|_2^2
	\le 1 .
	\]
	Therefore,
	\[
	\begin{aligned}
		& \mathbb P\!\left[\|S_n-n\bar P_n\|_2>K_\alpha\sqrt n\right]\\
		& \le
		\frac{\mathbb E\|S_n-n\bar P_n\|_2^2}{K_\alpha^2 n} \\
		& =
		\frac{\operatorname{tr}\operatorname{Cov}(S_n)}{K_\alpha^2 n}\\
		& \le
		\frac1{K_\alpha^2}.
	\end{aligned}
	\]
	Choosing \(K_\alpha\ge \sqrt{2/\alpha}\) gives the first claim.
	
	For the pointwise upper bound, project the count vector onto its first
	\(m-1\) coordinates and use Fourier inversion on \([-\pi,\pi]^{m-1}\).
	For one summand with law \(P\), set \(\theta_m=0\) and write
	\[
	\varphi_P(\theta)
	=
	\sum_{y=1}^{m-1}P(y)e^{i\theta_y}+P(m).
	\]
	Then
	\[
	|\varphi_P(\theta)|^2
	=
	1-2A_P(\theta),
	\]
	where
	\[
	A_P(\theta):=
	\sum_{1\le i<j\le m}
	P(i)P(j)\bigl(1-\cos(\theta_i-\theta_j)\bigr).
	\]
	Since \(0\le 2A_P(\theta)\le1\), the inequality \(1-x\le e^{-x}\) gives
	\[
	|\varphi_P(\theta)|\le e^{-A_P(\theta)}.
	\]
	Moreover, if
	\[
	\Psi(\theta):=
	\sum_{1\le i<j\le m}
	\bigl(1-\cos(\theta_i-\theta_j)\bigr),
	\qquad \theta_m=0,
	\]
	then \(A_P(\theta)\ge p_{\min}^2\Psi(\theta)\).  Since \(\Psi\) contains the
	pairs \((i,m)\), \(1\le i\le m-1\), and
	\(1-\cos x\ge 2x^2/\pi^2\) for \(|x|\le\pi\),
	\[
	\Psi(\theta)
	\ge
	\sum_{i=1}^{m-1}(1-\cos\theta_i)
	\ge
	\frac{2}{\pi^2}\|\theta\|_2^2 .
	\]
	Hence, uniformly over all \(P\) with \(P(y)\ge p_{\min}\),
	\[
	|\varphi_P(\theta)|
	\le
	\exp\left\{
	-\frac{2p_{\min}^2}{\pi^2}\|\theta\|_2^2
	\right\}.
	\]
	For the independent sum \(S_n\), the characteristic function of the projected
	count vector is the product of the characteristic functions of the projected
	summands.  Thus, with \(c_0=2p_{\min}^2/\pi^2\), Fourier inversion gives
	\[
	\begin{aligned}
		\sup_{\nu}\mathbb P[S_n=\nu]
		&\le
		\frac{1}{(2\pi)^{m-1}}
		\int_{[-\pi,\pi]^{m-1}}
		e^{-c_0n\|\theta\|_2^2}\,d\theta  \\
		&\le
		c_2 n^{-(m-1)/2},
	\end{aligned}
	\]
	where the supremum is over all count vectors
	\(\nu\in\mathbb Z_{\ge0}^m\) with \(\sum_y\nu_y=n\).
	
	Finally, the multinomial lower bound follows from Stirling's formula.  Let
	\(\hat q=\nu/n\).  If \(\|\nu-nq\|_2\le C_0\sqrt n\), then
	\(\|\hat q-q\|_2\le C_0/\sqrt n\).  Since \(q(y)\ge p_{\min}\), for all
	sufficiently large \(n\) we have \(\hat q(y)\ge p_{\min}/2\) for every \(y\).
	Stirling's formula, uniformly over all such \(q\) and \(\nu\), gives constants
	\(c_3>0\) and \(C<\infty\), depending only on \((p_{\min},m)\), such that
	\[
	\operatorname{Mult}(n,q)(\nu)
	\ge
	c_3 n^{-(m-1)/2}\,2^{-nD(\hat q\|q)} .
	\]
	Moreover, by upper bounding the KL divergence by the \(\chi^2\)-divergence,
	\[
	D(\hat q\|q)
	\le
	C\|\hat q-q\|_2^2
	\le
	\frac{C C_0^2}{n}.
	\]
	Therefore,
	\[
	\operatorname{Mult}(n,q)(\nu)
	\ge
	c_3 2^{-C C_0^2} n^{-(m-1)/2}.
	\]
	Setting \(c_1:=c_3 2^{-C C_0^2}\) proves the claimed multinomial lower bound.

	\section{Proof of Lemma \ref{lem:uniform-moments}} \label{apx_proof_uniform_moment}
	The upper bounds are immediate because \(\cY\) is finite, \(K_W\) is
	compact, and \(L\) is fixed.  It remains to prove that the variances are
	uniformly bounded away from zero.
	
	Fix an ordered pair \((i,j)\in\mathcal I_d\), and define the linear
	functional
	\[
	\ell_{ij}(z):=\langle Lz,e_i-e_j\rangle,
	\qquad z\in\bbR^m .
	\]
	Suppose, for contradiction, that \(V_{ij}(u)=0\) for some \(u\in K_W\).
	Then \(Z_{ij,u}(Y)\) is constant \(q_u\)-almost surely. Since
	\(\mathbb E_{q_u}[Z_{ij,u}(Y)]=0\), this constant must be zero. By
	Assumption~\ref{ass:positive-support}, \(q_u(y)>0\) for every \(y\in\cY\),
	and hence
	\[
	Z_{ij,u}(y)
	=
	\ell_{ij}(b_y-q_u)
	=
	0,
	\qquad \forall y\in\cY .
	\]
	Equivalently,
	\[
	\ell_{ij}(b_y)=\ell_{ij}(q_u),
	\qquad \forall y\in\cY .
	\]
	Thus there exists a constant \(c_0=\ell_{ij}(q_u)\) such that
	\[
	\ell_{ij}(b_y)=c_0,
	\qquad \forall y\in\cY .
	\]
	Therefore, for every \(z\in\bbR^m\),
	\[
	\ell_{ij}(z)
	=
	\sum_{y\in\cY} z(y)\ell_{ij}(b_y)
	=
	c_0\sum_{y\in\cY}z(y).
	\]
	In particular, \(\ell_{ij}\) is constant on the affine hyperplane
	\[
	\mathcal A_1
	:=
	\left\{z\in\bbR^m:\sum_{y\in\cY}z(y)=1\right\}.
	\]
	Since \(S_W^\star\subseteq\aff(\cP_W)\subseteq\mathcal A_1\),
	\(\ell_{ij}\) is constant on \(S_W^\star\).
	
	On the other hand, by the definition \(H=T^{-1}\circ\Pi_A\), the map \(H\)
	coincides with \(T^{-1}\) on \(A=\aff(\cP_W)\). Since
	\(S_W^\star\subseteq A\), for \(q=T(v)\in S_W^\star\) we have \(H(q)=v\).
	Writing \(H(q)=Lq+b\), we obtain
	\[
	\ell_{ij}(q)
	=
	\langle Lq,e_i-e_j\rangle
	=
	v_i-v_j-\langle b,e_i-e_j\rangle .
	\]
	As \(v\) ranges over \(\Delta_d\), the quantity \(v_i-v_j\) is not
	constant. Therefore \(\ell_{ij}\) cannot be constant on \(S_W^\star\),
	which contradicts the previous conclusion.
	
	Thus \(V_{ij}(u)>0\) for every \(u\in K_W\) and every
	\((i,j)\in\mathcal I_d\). Since \(V_{ij}(u)\) is continuous in \(u\),
	\(K_W\) is compact, and \(\mathcal I_d\) is finite, we obtain
	\[
	\min_{u\in K_W}\min_{(i,j)\in\mathcal I_d} V_{ij}(u)>0.
	\]
	This gives the desired uniform lower bound \(V_{\min}>0\).

	\section{A Uniform Lattice Local CLT for Count Vectors}
	\label{app:uniform-lattice-local-clt}
	
	We record a finite-alphabet specialization of the triangular-array lattice
	local CLT needed in Lemma~\ref{lem:jeffreys-local-count-lr}.  The result is a
	standard consequence of the Fourier-inversion proof of lattice local
	expansions; see \cite[Ch.~5, Sec.~22]{bhattacharya_rao_normal_1976} and the
	non-i.i.d. extension described in \cite[pp.~240--241]{bhattacharya_rao_normal_1976}. We include the proof here for completeness.
	
	Let \(r\ge1\), let \(e_0:=0\in\bbR^r\), and let
	\[
	E_r:=\{e_0,e_1,\ldots,e_r\}\subset\bbZ^r.
	\]
	
	\begin{lemma}
		\label{lem:appendix-uniform-categorical-local-clt}
		Fix \(\rho>0\).  For each \(n\), let
		\(X_{n,1},\ldots,X_{n,n}\) be independent random vectors taking values in
		\(E_r\).  Assume
		\[
		\mathbb P[X_{n,t}=e_j]\ge \rho,
		\qquad
		0\le j\le r,\ 1\le t\le n .
		\]
		Let
		\[
		\mu_{n,t}:=\mathbb E X_{n,t},
		\qquad
		V_n:=\frac1n\sum_{t=1}^n \operatorname{Cov}(X_{n,t}),
		\qquad
		S_n:=\sum_{t=1}^n X_{n,t}.
		\]
		Then, for every compact set \(K\subset\bbR^r\),
		\[
		\sup
		\left|
		\frac{
			(2\pi n)^{r/2}\sqrt{\det V_n}\,
			\mathbb P[S_n=a]
		}{
			\exp\{-\frac12 w_{n,a}^{\top}V_n^{-1}w_{n,a}\}
		}
		-1
		\right|
		\to0,
		\]
		where the supremum is over all such triangular arrays and all
		\(a\in\bbZ^r\) satisfying
		\[
		w_{n,a}:=
		\frac{a-\sum_{t=1}^n\mu_{n,t}}{\sqrt n}
		\in K .
		\]
	\end{lemma}
	
	\begin{proof}
		All constants below depend only on \(\rho\) and \(r\).  First note that the
		covariances are uniformly nondegenerate.  Indeed, for any unit vector
		\(\theta\in\bbR^r\), the random variable \(\theta^\top X_{n,t}\) takes the
		values \(0,\theta_1,\ldots,\theta_r\), each with probability at least \(\rho\).
		Since \(\max_j|\theta_j|\ge r^{-1/2}\), the range of these values is at least
		\(r^{-1/2}\).  Using
		\[
		\operatorname{Var}(U)
		=
		\frac12\sum_{i,j}p_i p_j(u_i-u_j)^2,
		\]
		we obtain
		\[
		\theta^\top \operatorname{Cov}(X_{n,t})\theta\ge c_-
		\]
		for some \(c_->0\).  The upper bound is immediate from boundedness of
		\(E_r\).  Hence
		\[
		c_- I_r\preceq V_n\preceq c_+ I_r
		\]
		uniformly in the array.
		
		Let
		\[
		\psi_{n,t}(\theta)
		:=
		\mathbb E\exp\{i\theta^\top(X_{n,t}-\mu_{n,t})\}.
		\]
		Fourier inversion on the lattice \(\bbZ^r\) gives
		\[
		\mathbb P[S_n=a]
		=
		(2\pi)^{-r}
		\int_{[-\pi,\pi]^r}
		\exp\!\left\{
		-i\theta^\top\left(a-\sum_{t=1}^n\mu_{n,t}\right)
		\right\}
		\prod_{t=1}^n\psi_{n,t}(\theta)
		d\theta .
		\]
		With the change of variables \(\theta=h/\sqrt n\),
		\[
		n^{r/2}\mathbb P[S_n=a]
		=
		(2\pi)^{-r}
		\int_{\sqrt n[-\pi,\pi]^r}
		e^{-ih^\top w_{n,a}}
		\prod_{t=1}^n\psi_{n,t}(h/\sqrt n)
		dh .
		\]
		
		On every fixed ball \(\|h\|\le M\), Taylor expansion gives, uniformly over
		the array,
		\[
		\log \psi_{n,t}(h/\sqrt n)
		=
		-\frac1{2n}h^\top\operatorname{Cov}(X_{n,t})h
		+
		O\!\left(\frac{\|h\|^3}{n^{3/2}}\right).
		\]
		Therefore
		\[
		\prod_{t=1}^n\psi_{n,t}(h/\sqrt n)
		=
		\exp\left\{-\frac12 h^\top V_nh\right\}(1+o(1))
		\]
		uniformly for \(\|h\|\le M\).
		
		It remains to control the tails uniformly.  For sufficiently small
		\(\delta>0\), the same Taylor expansion and the uniform lower bound on the
		covariance imply
		\[
		|\psi_{n,t}(\theta)|
		\le
		\exp\{-c\|\theta\|^2\},
		\qquad
		\|\theta\|\le\delta .
		\]
		Hence, for \(M\le\|h\|\le\delta\sqrt n\),
		\[
		\prod_{t=1}^n|\psi_{n,t}(h/\sqrt n)|
		\le
		\exp\{-c\|h\|^2\}.
		\]
		For the remaining region
		\(\delta\le\|\theta\|\le \pi\sqrt r\), the common minimal lattice property
		gives a uniform aperiodicity gap.  Indeed, the probability vector
		\((\mathbb P[X=e_j])_{j=0}^r\) ranges over a compact subset of the simplex. The only points in \([-\pi,\pi]^r\) at which
		\(|\mathbb E e^{i\theta^\top X}|=1\) for all distributions supported on
		\(E_r\) are the points of \(2\pi\mathbb Z^r\), hence only \(0\) in the chosen
		fundamental domain.  Since the probabilities are restricted to the compact set
		\(\{p_j\ge\rho,\sum_j p_j=1\}\), the supremum over
		\(\{\delta\le\|\theta\|\le\pi\sqrt r\}\) is strictly smaller than one.  Thus
		\[
		\sup_{\substack{\delta\le\|\theta\|\le \pi\sqrt r\\
				\mathbb P[X=e_j]\ge\rho}}
		\left|
		\mathbb E e^{i\theta^\top X}
		\right|
		<1 .
		\]
		Consequently the contribution of this region is exponentially small in \(n\).
		
		Combining the local approximation and the two tail estimates gives, uniformly
		for \(w_{n,a}\in K\),
		\[
		n^{r/2}\mathbb P[S_n=a]
		=
		(2\pi)^{-r}
		\int_{\bbR^r}
		e^{-ih^\top w_{n,a}}
		e^{-\frac12h^\top V_nh}
		dh
		+
		o(1).
		\]
		The integral equals
		\[
		(2\pi)^{-r/2}(\det V_n)^{-1/2}
		\exp\left\{
		-\frac12 w_{n,a}^{\top}V_n^{-1}w_{n,a}
		\right\}.
		\]
		Since \(w_{n,a}\in K\) and the eigenvalues of \(V_n\) are uniformly bounded
		above and below, this Gaussian factor is uniformly bounded away from zero.
		The asserted relative estimate follows.
	\end{proof}
	
	We next translate the coordinate version back to count vectors in
	\(\mathcal H=\{z\in\bbR^m:\sum_y z_y=0\}\).
	
	\begin{lemma}[Uniform local CLT for count vectors]
		\label{lem:appendix-uniform-count-local-clt}
		Fix \(\rho>0\).  Let \(Y_{n,1},\ldots,Y_{n,n}\) be independent
		\(\mathcal Y\)-valued random variables satisfying
		\[
		\mathbb P[Y_{n,t}=y]\ge\rho,
		\qquad y\in\mathcal Y,\ 1\le t\le n .
		\]
		Let \(N_n\) be the count vector of \(Y_{n,1},\ldots,Y_{n,n}\), and define
		\[
		\bar q_n:=\frac1n\sum_{t=1}^n \mathbb P[Y_{n,t}=\cdot],
		\]
		\[
		\bar\Sigma_n
		:=
		\frac1n\sum_{t=1}^n
		\left[
		\diag(r_{n,t})-r_{n,t}r_{n,t}^{\top}
		\right],
		\qquad
		r_{n,t}:=\mathbb P[Y_{n,t}=\cdot].
		\]
		Then, for every compact \(K\subset\mathcal H\), uniformly over all such
		arrays and all count vectors \(\nu\) satisfying
		\[
		z_{n,\nu}:=\frac{\nu-n\bar q_n}{\sqrt n}\in K,
		\]
		we have
		\[
		\mathbb P[N_n=\nu]
		=
		\frac{\Delta_{\mathcal H}(1+o(1))}
		{(2\pi n)^{(m-1)/2}\sqrt{\det_{\mathcal H}\bar\Sigma_n}}
		\exp\left\{
		-\frac12 z_{n,\nu}^{\top}\bar\Sigma_n^{-1}z_{n,\nu}
		\right\}.
		\]
		Here \(\Delta_{\mathcal H}\) is the fundamental volume of the count lattice
		\(\mathbb Z^m\cap\mathcal H\) in the Euclidean geometry of \(\mathcal H\).
	\end{lemma}
	
	\begin{proof}
		Fix a reference output symbol \(y_0\), and identify \(\mathcal H\) with
		\(\bbR^{m-1}\) by deleting the \(y_0\)-coordinate.  Let
		\(G:\bbR^{m-1}\to\mathcal H\) be the inverse linear map, so that \(G\) inserts
		the deleted coordinate by enforcing the zero-sum constraint.  The count lattice
		in \(\mathcal H\) is \(G\bbZ^{m-1}\), whose fundamental volume is
		\[
		\Delta_{\mathcal H}=\sqrt{\det(G^\top G)} .
		\]
		
		Apply Lemma~\ref{lem:appendix-uniform-categorical-local-clt} to the projected
		summands.  Let \(V_n\) be the covariance matrix of the projected summands.
		Then
		\[
		\bar\Sigma_n=G V_nG^\top
		\]
		as a covariance operator on \(\mathcal H\).  If determinants and inverses of
		\(\bar\Sigma_n\) are taken in the Euclidean geometry of \(\mathcal H\), then
		\[
		\det_{\mathcal H}\bar\Sigma_n
		=
		\det(G^\top G)\det V_n
		=
		\Delta_{\mathcal H}^2\det V_n,
		\]
		and for \(z=Gw\in\mathcal H\),
		\[
		z^\top\bar\Sigma_n^{-1}z
		=
		w^\top V_n^{-1}w .
		\]
		Substituting these identities into the coordinate local CLT proves the
		displayed formula.
	\end{proof}

	\section{Proof of Proposition \ref{prop:stratified-uniform-testing-bound}} \label{apx:proof_stratified-uniform-testing-bound}

	Let \(p_n:=\pi_{x^n}\) and \(q_n:=p_nW\).  Choose \(a>0\) small enough as
	specified below, and set
	\[
	\tau_n:=a\sqrt{\frac{\log n}{n}} .
	\]
	We split input sequences into the interior zone
	\[
	\operatorname{dist}(q_n,\relbd\mathcal P_W)\ge \tau_n
	\]
	and the boundary zone, where the reverse inequality holds.  Put
	\[
	a_n:=\frac d2\log n+\log\mathcal J_W-\frac d2\log(2\pi).
	\]
	
	First consider the interior zone.  Since
	\[
	Q_{Y^n}^{\rm str}\ge (1-\delta_n)Q_{\mathcal P_W,n}^{J},
	\]
	we have
	\[
	\begin{aligned}
		&-\log\beta_\alpha(P_{Y^n|x^n}^{\rm perm},Q_{Y^n}^{\rm str})\nonumber \\
		&\le
		-\log\beta_\alpha(P_{Y^n|x^n}^{\rm perm},Q_{\mathcal P_W,n}^{J})
		-\log(1-\delta_n),
	\end{aligned}
	\]
	and \(-\log(1-\delta_n)=o(1)\).  It is therefore enough, in the interior
	zone, to bound the testing term with \(Q_{\mathcal P_W,n}^{J}\).
	
	We claim that, uniformly over the interior zone,
	\[
	\begin{aligned}
		&-\log\beta_\alpha
		\left(
		P_{Y^n|x^n}^{\rm perm},
		Q_{\mathcal P_W,n}^{J}
		\right) \nonumber \\
		&\le
		a_n+
		\sup_{p\in\Delta_{|\mathcal X|-1}}[-\log B_\alpha(p)]
		+o(1).
	\end{aligned}
	\]
	Suppose not.  Then there exist \(\delta>0\), a subsequence \(n_\ell\to\infty\),
	and input sequences \(x^{n_\ell}\) in the interior zone such that, with
	\[
	P_\ell:=P_{Y^{n_\ell}|x^{n_\ell}}^{\rm perm},
	\quad
	Q_\ell:=Q_{\mathcal P_W,n_\ell}^{J},
	\quad
	p_\ell:=\pi_{x^{n_\ell}},
	\]
	we have
	\[
	-\log\beta_\alpha(P_\ell,Q_\ell)
	>
	a_{n_\ell}
	+
	\sup_{p\in\Delta_{|\mathcal X|-1}}[-\log B_\alpha(p)]
	+\delta
	\]
	for all \(\ell\).  By compactness, after passing to a further subsequence,
	\(p_\ell\to p\).  Let \(q_\ell=p_\ell W\) and \(q=pW\).
	
	Although \(q\) may lie on \(\relbd\mathcal P_W\), the interior-zone condition gives
	\[
	\sqrt{n_\ell}\operatorname{dist}(q_\ell,\relbd\mathcal P_W)
	\ge
	a\sqrt{\log n_\ell}\to\infty .
	\]
	Hence Lemma~\ref{lem:jeffreys-laplace-interior-zone} applies along this
	subsequence.  Together with Lemma~\ref{lem:jeffreys-local-count-lr}, and using
	the tightness of
	\[
	Z_{n_\ell}:=\frac{N(Y^{n_\ell})-n_\ell q_\ell}{\sqrt{n_\ell}},
	\]
	we obtain
	\[
	\log\frac{dP_\ell}{dQ_\ell}(Y^{n_\ell})-a_{n_\ell}
	=
	H_{p_\ell}(Z_{n_\ell})+o_{P_\ell}(1).
	\]
	The local CLT gives
	\[
	Z_{n_\ell}\Rightarrow Z_p\sim\mathcal N(0,\Sigma(p)).
	\]
	Moreover \(H_{p_\ell}\to H_p\) uniformly on compact subsets of \(\mathcal H\).
	Therefore
	\[
	H_{p_\ell}(Z_{n_\ell})\Rightarrow H_p(Z_p),
	\]
	and hence
	\[
	\log\frac{dP_\ell}{dQ_\ell}(Y^{n_\ell})-a_{n_\ell}
	\Rightarrow
	H_p(Z_p).
	\]
	
	Let
	\[
	W_\ell
	:=
	2^{-\left(
		\log\frac{dP_\ell}{dQ_\ell}(Y^{n_\ell})-a_{n_\ell}
		\right)},
	\qquad
	W:=2^{-H_p(Z_p)} .
	\]
	Then \(W_\ell\Rightarrow W\).  For a nonnegative random variable \(X\), let
	\[
	\mathsf B_\alpha(X)
	:=
	\inf_{\substack{\varphi:\mathbb R_+\to[0,1]\\
			\mathbb E[\varphi(X)]\ge\alpha}}
	\mathbb E[X\varphi(X)] .
	\]
	By the Neyman-Pearson lemma, it is enough to consider tests measurable
	with respect to the likelihood-ratio statistic.  For a nonnegative random variable \(X\), the corresponding lower-tail
	functional admits the variational representation
	\[
	\mathsf B_\alpha(X)
	=
	\sup_{\lambda\ge0}
	\left\{
	\alpha\lambda-\mathbb E[(\lambda-X)_+]
	\right\}.
	\]
	This variational identity is the lower-tail form of the CVaR variational
	formula \cite[Theorem~1]{rockafellar_uryasev_cvar_2000}. It is obtained by
	applying that formula to the loss \(-X\) at confidence level \(1-\alpha\),
	with the standard randomized-threshold interpretation when atoms are present.  Since
	\(W_\ell\Rightarrow W\), for each fixed \(\lambda\ge0\),
	\[
	\mathbb E[(\lambda-W_\ell)_+]\to
	\mathbb E[(\lambda-W)_+],
	\]
	because \(x\mapsto(\lambda-x)_+\) is bounded and continuous on
	\(\mathbb R_+\).  The equality \(\mathsf B_\alpha(W)=B_\alpha(p)\) follows by conditioning:
	for any admissible test \(\varphi(Z_p)\), replacing it by
	\(\mathbb E[\varphi(Z_p)\mid W]\) preserves its power and its value of
	\(\mathbb E[W\varphi(Z_p)]\).  Conversely, every test measurable with respect
	to \(W=2^{-H_p(Z_p)}\) is also an admissible test measurable with respect to
	\(Z_p\). Therefore
	\[
	\liminf_{\ell\to\infty}\mathsf B_\alpha(W_\ell)
	\ge
	\mathsf B_\alpha(W)
	=
	B_\alpha(p).
	\]
	Since
	\[
	\frac{dQ_\ell}{dP_\ell}(Y^{n_\ell})
	=
	2^{-a_{n_\ell}}W_\ell,
	\]
	the Neyman-Pearson reduction gives
	\[
	\beta_\alpha(P_\ell,Q_\ell)
	=
	2^{-a_{n_\ell}}\mathsf B_\alpha(W_\ell).
	\]
	We then obtain
	\[
	\limsup_{\ell\to\infty}
	\left[
	-\log\beta_\alpha(P_\ell,Q_\ell)-a_{n_\ell}
	\right]
	\le
	-\log B_\alpha(p).
	\]
	This contradicts the preceding strict inequality, because
	\[
	-\log B_\alpha(p)
	\le
	\sup_{\tilde p\in\Delta_{|\mathcal X|-1}}[-\log B_\alpha(\tilde p)] .
	\]
	Thus the claimed interior-zone bound holds.  Combining it with
	\(Q_{Y^n}^{\rm str}\ge(1-\delta_n)Q_{\mathcal P_W,n}^{J}\) gives
	\[
	-\log\beta_\alpha
	\left(
	P_{Y^n|x^n}^{\rm perm},
	Q_{Y^n}^{\rm str}
	\right)
	\le
	\frac d2\log n+\log\mathcal J_W+C_{W,\alpha}^{J,{\rm int}}+o(1)
	\]
	uniformly over the interior zone.
	
	It remains to treat the boundary zone.  Let
	\[
	\mathcal A_n:=\{\|N(Y^n)-nq_n\|\le L\sqrt n\},
	\]
	where \(L\) is chosen so that
	\[
	P_{Y^n|x^n}^{\rm perm}(\mathcal A_n)\ge1-\frac{\alpha}{2}
	\]
	uniformly in \(x^n\).  This follows from the uniform second-moment bound
	\(\mathbb E\|N(Y^n)-nq_n\|^2\le C_W n\).  On \(\mathcal A_n\),
	Lemma~\ref{lem:count-space-likelihood-ratio} together with
	Lemma~\ref{lem:uniform-count-local} gives
	\[
	\frac{P_{Y^n|x^n}^{\rm perm}(y^n)}
	{q_n^{\otimes n}(y^n)}
	\le C_1
	\]
	uniformly.
	
	Now suppose \(q_n\) is in the boundary zone.  Choose
	\(\bar q_n\in\relbd\mathcal P_W\) with
	\[
	\|q_n-\bar q_n\|\le\tau_n .
	\]
	Let \(F_n\in\mathfrak F_\partial\) be the unique proper face whose relative
	interior contains \(\bar q_n\), and write \(r_n:=\dim F_n\le d-1\).
	Since \(\mathcal P_W\) is a polytope with finitely many faces, its face lattice
	has a positive minimum relative solid angle.  Equivalently, there exist
	constants \(c_{\rm face}>0\) and \(\rho_0>0\), depending only on
	\(\mathcal P_W\), such that for every proper face \(F\), every
	\(\bar q\in F\), and every \(0<\rho\le\rho_0\),
	\[
	\operatorname{vol}_{\dim F}\bigl(F\cap B(\bar q,\rho)\bigr)
	\ge
	c_{\rm face}\rho^{\dim F}.
	\]
	Applying this with \(F=F_n\) and \(\rho=\rho_0/\sqrt n\), and absorbing
	\(\rho_0^{r_n}\) into the constant, gives
	\[
	\operatorname{vol}_{r_n}
	\left(F_n\cap B(\bar q_n,\rho_0/\sqrt n)\right)
	\ge
	c_F n^{-r_n/2}.
	\]
	For \(r_n=0\), we use the convention
	\(\operatorname{vol}_0(\{\bar q_n\})=1\).  For positive-dimensional faces, the
	normalized face Jeffreys densities are uniformly bounded below; vertex faces
	are handled by the convention \(Q_{F,n}^{J}=q_F^{\otimes n}\).
	
	For \(v\in F_n\cap B(\bar q_n,\rho_0/\sqrt n)\),
	\[
	\|v-q_n\|
	\le
	\tau_n+\rho_0/\sqrt n
	\le
	2a\sqrt{\frac{\log n}{n}}
	\]
	for all large \(n\).  Write \(\delta_y:=v(y)-q_n(y)\).  On \(\mathcal A_n\),
	write
	\[
	N_y=nq_n(y)+\Delta_y,
	\qquad
	\|\Delta\|\le L\sqrt n .
	\]
	Since \(q_n(y)\ge p_{\min}\) and
	\(\|\delta\|\le 2a\sqrt{\log n/n}\), for all sufficiently large \(n\),
	\[
	\left|\frac{\delta_y}{q_n(y)}\right|\le \frac12,
	\qquad y\in\mathcal Y .
	\]
	Using the Taylor lower bound
	\(\log(1+t)\ge (\log e)t-C_Tt^2\) for \(|t|\le1/2\), we obtain
	\[
	\begin{aligned}
		\sum_y N_y\log\frac{v(y)}{q_n(y)}
		&=
		\sum_y (nq_n(y)+\Delta_y)
		\log\left(1+\frac{\delta_y}{q_n(y)}\right)\\
		&\ge
		(\log e)n\sum_y\delta_y
		-
		C_{\rm T,1} n\|\delta\|^2
		-
		C_{\rm T,2}\|\Delta\|\|\delta\|.
	\end{aligned}
	\]
	The first term is zero because both \(v\) and \(q_n\) are probability
	distributions.  On \(\mathcal A_n\), \(\|\Delta\|\le L\sqrt n\), and hence
	\[
	\sum_y N_y\log\frac{v(y)}{q_n(y)}
	\ge
	-C_{\rm T,1} n\|\delta\|^2-C_{\rm T,2}L\sqrt n\,\|\delta\|.
	\]
	Using \(\|\delta\|\le 2a\sqrt{\log n/n}\), we get
	\[
	\sum_y N_y\log\frac{v(y)}{q_n(y)}
	\ge
	-4C_{\rm T,1}a^2\log n
	-
	2C_{\rm T,2}La\sqrt{\log n}.
	\]
	Since \(a>0\) is fixed, the second term is \(o(\log n)\).  Thus, for a
	constant \(K<\infty\) depending only on \(W,L\), and the Taylor constants,
	\[
	\sum_y N_y\log\frac{v(y)}{q_n(y)}
	\ge
	-Ka^2\log n+o(\log n)
	\]
	uniformly over the boundary zone. Hence
	\[
	v^{\otimes n}(y^n)
	\ge
	n^{-Ka^2+o(1)}q_n^{\otimes n}(y^n).
	\]
	Integrating over \(F_n\cap B(\bar q_n,\rho_0/\sqrt n)\) gives
	\[
	Q_{F_n,n}^{J}(y^n)
	\ge
	c_{\rm J} n^{-r_n/2-Ka^2+o(1)}
	q_n^{\otimes n}(y^n).
	\]
	Since
	\[
	Q_{Y^n}^{\rm str}
	\ge
	\delta_n\omega_{F_n}Q_{F_n,n}^{J},
	\]
	we obtain
	\[
	Q_{Y^n}^{\rm str}(y^n)
	\ge
	c_{\rm str} n^{-\gamma-r_n/2-Ka^2+o(1)}
	q_n^{\otimes n}(y^n).
	\]
	Together with the count likelihood-ratio bound, this yields on \(\mathcal A_n\)
	\[
	\frac{P_{Y^n|x^n}^{\rm perm}(y^n)}
	{Q_{Y^n}^{\rm str}(y^n)}
	\le
	C_{\rm LR} n^{r_n/2+\gamma+Ka^2+o(1)} .
	\]
	Choose \(a>0\) so small that
	\[
	\gamma+Ka^2<\frac12 .
	\]
	Since \(r_n\le d-1\), there exists \(\kappa>0\) such that
	\[
	r_n/2+\gamma+Ka^2
	\le
	\frac d2-\kappa .
	\]
	Thus, on \(\mathcal A_n\),
	\[
	\frac{dP}{dQ}(y^n)
	\le
	C_{\rm LR} n^{d/2-\kappa+o(1)},
	\]
	where \(P:=P_{Y^n|x^n}^{\rm perm}\) and \(Q:=Q_{Y^n}^{\rm str}\).
	
	Let \(0\le\varphi\le1\) be any randomized test with
	\(\mathbb E_P\varphi\ge\alpha\).  Since \(P(\mathcal A_n)\ge1-\alpha/2\),
	\[
	\mathbb E_P[\varphi\mathbf 1_{\mathcal A_n}]
	\ge
	\frac{\alpha}{2}.
	\]
	Using the likelihood-ratio bound on \(\mathcal A_n\),
	\[
	\mathbb E_Q[\varphi]
	\ge
	C_{\rm LR}^{-1}n^{-d/2+\kappa+o(1)}
	\mathbb E_P[\varphi\mathbf 1_{\mathcal A_n}]
	\ge
	c_\alpha n^{-d/2+\kappa+o(1)} .
	\]
	Taking the infimum over all such tests gives
	\[
	\beta_\alpha(P,Q)
	\ge
	c_\alpha n^{-d/2+\kappa+o(1)}.
	\]
	Equivalently,
	\[
	-\log\beta_\alpha(P,Q)
	\le
	\left(\frac d2-\kappa\right)\log n+o(\log n)
	\]
	uniformly over the boundary zone.
	
	Therefore, after subtracting \(\frac d2\log n\), the boundary-zone contribution
	has limsup \(-\infty\), while the interior zone contributes at most
	\[
	\log\mathcal J_W+C_{W,\alpha}^{J,{\rm int}} .
	\]
	This proves the proposition.

	\bibliographystyle{IEEEtran}
	\bibliography{low_rank_gaussian_extension_refs}

	\begin{IEEEbiographynophoto}{Lugaoze Feng}
		received the B.S. degree from Xidian University, Xi'an, China, in 2023. He is currently pursuing the Ph.D. degree in communication and information systems with Peking University, Beijing. His research interests include information theory and channel coding.
	\end{IEEEbiographynophoto}
	
	\begin{IEEEbiographynophoto}{Guocheng Lv}
		received the B.S. degree from Peking University, Beijing, China, in 2006, and the M.S. degree from Peking University, Beijing, China, in 2009. He is currently a Senior Engineer with the School of Electronics, Peking University. His research interests include satellite communication, physical layer modem and non-orthogonal multiple access.
	\end{IEEEbiographynophoto}
	
	\begin{IEEEbiographynophoto}{Xunan Li}
		received the B.S. degree in Telecommunications Engineering from Nankai University, Tianjin, China, in 2013, and the Ph.D. degree in Communications and Information System from Peking University, Beijing, China, in 2018. His research interests include communication signal processing and Satellite Communications.
	\end{IEEEbiographynophoto}
	
	\begin{IEEEbiographynophoto}{Ye Jin}
		received the B.E. and M.S. degrees from Peking University, Beijing, China, in 1986 and 1989, respectively. He is currently a Professor with the Institute of Modern Communications, Peking University. He has been the Principal Investigator of over 30 funded research projects. His general research interests are in the areas of satellite and wireless communications and networking. Prof. Jin was a recipient of the First Prize of the National Science and Technology Progress Awards of China.
	\end{IEEEbiographynophoto}

\end{document}